\documentclass[journal,twoside,web]{ieeecolor}
\usepackage{generic}
\usepackage{cite}
\usepackage{amsmath,amssymb,amsfonts}
\usepackage{graphicx}
\usepackage{textcomp}

\usepackage{mathrsfs}
\usepackage{mathtools}
\usepackage{bm}
\usepackage{cite}
\usepackage{url}
\usepackage{fancybox}
\usepackage{algorithm}
\usepackage{algorithmicx}
\usepackage{algpseudocode}
\usepackage{subfigure}
\usepackage{multirow}

\newcommand{\figref}[1]{Fig.~\ref{#1}}
\newcommand{\tabref}[1]{Table~\ref{#1}}
\newcommand{\argref}[1]{Algorithm~\ref{#1}}
\newcommand{\defref}[1]{Definition~\ref{#1}}
\newcommand{\probref}[1]{Problem~\ref{#1}}
\newcommand{\propref}[1]{Proposition~\ref{#1}}
\newcommand{\propsref}[2]{Propositions~\ref{#1} and \ref{#2}}
\newcommand{\lemref}[1]{Lemma~\ref{#1}}
\newcommand{\thmref}[1]{Theorem~\ref{#1}}
\newcommand{\corref}[1]{Corollary~\ref{#1}}

\newcommand{\appref}[1]{Appendix~\ref{#1}}
\newcommand{\secref}[1]{Section~\ref{#1}}

\newcommand{\F}{\mathbb{F}}
\newcommand{\G}{\mathbb{G}}
\newcommand{\R}{\mathbb{R}}

\newcommand{\Z}{\mathbb{Z}}

\renewcommand{\S}{\mathcal{S}}
\newcommand{\K}{\mathcal{K}}
\newcommand{\Kp}{\mathcal{K}_{\mathsf{p}}}
\newcommand{\Ks}{\mathcal{K}_{\mathsf{s}}}

\newcommand{\E}{\mathcal{E}}
\newcommand{\Edyn}{\mathcal{E}_{\mathrm{dyn}}}
\newcommand{\M}{\mathcal{M}}
\newcommand{\C}{\mathcal{C}}
\newcommand{\Gen}{\mathsf{Gen}}
\newcommand{\Enc}{\mathsf{Enc}}
\newcommand{\Dec}{\mathsf{Dec}}
\newcommand{\pk}{\mathsf{pk}}
\newcommand{\sk}{\mathsf{sk}}
\newcommand{\Ecd}{\mathsf{Ecd}}
\newcommand{\Dcd}{\mathsf{Dcd}}
\newcommand{\Qtz}{\mathcal{Q}}

\newcommand{\tr}{\mathop{\mathrm{tr}}\limits}
\renewcommand{\vec}{\mathop{\mathrm{vec}}\limits}
\newcommand{\diag}[1]{\mathop{\mathrm{diag}}\limits\left(#1\right)}
\newcommand{\const}{\mathrm{const.}}
\newcommand{\argmin}{\mathop{\mathrm{arg~min}}\limits}

\newcommand{\ind}[2]{\mathbf{1}_{#1}\left(#2\right)}
\newcommand{\U}{\mathcal{U}}
\newcommand{\N}{\mathcal{N}}
\newcommand{\EV}[1]{\mathbb{E}\left[#1\right]}

\newcommand{\D}{\mathcal{D}}
\newcommand{\Denc}{\mathcal{D}_{\Enc}}
\newcommand{\bfzero}{\mathbf{0}}

\newtheorem{definition}{Definition}

\newtheorem{proposition}{Proposition}
\newtheorem{lemma}{Lemma}
\newtheorem{theorem}{Theorem}
\newtheorem{corollary}{Corollary}
\newtheorem{problem}{Problem}
\newtheorem{remark}{Remark}

\makeatletter
\let\MYcaption\@makecaption
\makeatother
\usepackage{subfigure}
\makeatletter
\let\@makecaption\MYcaption
\makeatother

\def\BibTeX{{\rm B\kern-.05em{\sc i\kern-.025em b}\kern-.08em
    T\kern-.1667em\lower.7ex\hbox{E}\kern-.125emX}}
\begin{document}

\thispagestyle{empty}
\hspace{-4.5mm}
\fbox{
\begin{minipage}{\textwidth-5mm}\scriptsize
© 20XX IEEE.  Personal use of this material is permitted.  Permission from IEEE must be obtained for all other uses, in any current or future media, including reprinting/republishing this material for advertising or promotional purposes, creating new collective works, for resale or redistribution to servers or lists, or reuse of any copyrighted component of this work in other works.
\end{minipage}
}
\newpage
\setcounter{page}{0}

\title{Designing Optimal Key Lengths and \\ Control Laws for Encrypted Control Systems based on Sample Identifying Complexity \\ and Deciphering Time}

\author{%
    Kaoru Teranishi, \IEEEmembership{Student Member, IEEE},
    Tomonori Sadamoto, \IEEEmembership{Member, IEEE}, \\
    Aranya Chakrabortty, \IEEEmembership{Senior Member, IEEE},
    and Kiminao Kogiso, \IEEEmembership{Member, IEEE}
    \thanks{%
        This work was supported by JSPS Grant-in-Aid for JSPS Fellows Grant Number JP21J22442.
    }
    \thanks{%
        Kaoru Teranishi, Tomonori Sadamoto, and Kiminao Kogiso are with the Department of Mechanical and Intelligent Systems Engineering, The University of Electro-Communications, Chofu, Tokyo, 182-8585, Japan (e-mail: teranishi@uec.ac.jp, sadamoto@uec.ac.jp, kogiso@uec.ac.jp).
    }
    \thanks{%
        Kaoru Teranishi is also a Research Fellow of Japan Society for the Promotion of Science.
    }
    \thanks{%
        Aranya Chakrabortty is with the Department of Electrical and Computer Engineering, North Carolina State University, Raleigh, NC 27695 USA (e-mail: achakra2@ncsu.edu).
    }
}

\maketitle

\begin{abstract}

In the state-of-the-art literature on cryptography and control theory, there has been no systematic methodology of constructing cyber-physical systems that can achieve desired control performance while being protected against eavesdropping attacks.
In this paper, we tackle this challenging problem.
We first propose two novel notions referred to as \textit{sample identifying complexity} and \textit{sample deciphering time} in an encrypted-control framework.
The former explicitly captures the relation between the dynamical characteristics of control systems and the level of identifiability of the systems while the latter shows the relation between the computation time for the identification and the key length of a cryptosystem.
Based on these two tractable new notions, we propose a systematic method for designing the both of an optimal key length to prevent system identification with a given precision within a given life span of systems, and of an optimal controller to maximize both of the control performance and the difficulty of the identification.
The efficiency of the proposed method in terms of security level and realtime-ness is investigated through numerical simulations.
To the best of our knowledge, this paper first connect the relationship between the security of cryptography and dynamical systems from a control-theoretic perspective. 

\end{abstract}


\section{Introduction}

\subsection{Motivational literature review}

Cyber-physical systems have attracted the attention in numerous areas, such as power grids, transportation, manufacturing, and healthcare~\cite{Lee08,Humayed17,Dibaji19}.
Integrating communication and computation layers with a physical layer, cyber-physical systems are expected to overwhelm the traditional systems with respect to efficiency, reliability, and sustainability~\cite{Lee08,Wang15}.
Meanwhile, cyber-physical systems often face security threats in exchange for the advantages because, in general, they communicate with a public and untrustworthy computer, e.g., cloud, over insecure channels for decision making.

One of major security threats is the eavesdropping attack that tries to disclose confidential information of cyber-physical systems~\cite{Teixeira15_1}.
Once an adversary complete the attacks, more destructive and undetectable attacks can be designed based on a target system model learned by the disclosed information~\cite{Chong19}.
Therefore, it is crucial for realizing secure cyber-physical systems to prevent eavesdropping attacks.

To fulfill this objective, we definitely need a \textit{measure} for quantifying the security level against the attacks.
Some studies have employed information-theoretic measures, such as mutual information and directed information, for designing estimators and controllers with information leakage constraints under the presence of eavesdroppers~\cite{Nekouei19}.
Additionally, differential privacy~\cite{Dwork06}, another well-known measure used in information community, has been adopted for private filtering and controls of dynamical systems~\cite{Cortes16,Hassan20}.
However, these existing measures are not suitable for \textit{dynamical} systems because it is not clear that the systems should satisfy how the level of security.
Furthermore, a controller design method based on the measures has an intrinsic trade-off between the security and quality of controls due to noise injection~\cite{Nekouei19,Cortes16}.
It should be noted here that some recent papers have proposed control-theoretic security quantities~\cite{Dibaji19,Murguia20,Sandberg10,Milosevic20,Feng21,Cetinkaya20}.
However, the quantities cannot measure the security level against eavesdropping attacks because they focus on other attacks.

Encrypted control~\cite{Darup20_3} is the state-of-the-art technology for preventing eavesdropping attacks without noise injection.
Contrary to the information-oriented methods, the performance degradation in encrypted control systems can be ignored by increasing a key length of cryptosystem~\cite{Kogiso18_1}.
Moreover, for a small key length, appropriate quantizers mitigate the quantization errors due to encryption~\cite{Kishida19,Teranishi19_3}.
Thus, encrypted control is a promising framework for achieving the superior security and control performance of cyber-physical systems.
In fact, various encrypted control methods have been developed recently by using partially, somewhat, and (leveled) fully homomorphic encryption~\cite{Kogiso15,Farokhi17,Kim16,Darup18_2,Darup19_1,Alexandru20_3,Alexandru20_2,Fritz19,Ristic20,Suh21}.
Moreover, their feasibility has been verified through implementation to a drone~\cite{Cheon18_1}, fog-computing environment~\cite{Teranishi20_4}, and field-programmable gate array~\cite{Tran20}.
However, the security level of encrypted control systems has not been analyzed and quantified.

\subsection{Contribution}

This study considers an attack scenario that an adversary eavesdrops and then identifies the system matrix of a stochastic closed-loop system with an encrypted controller by using collected encrypted-data.
Under this scenario, we aim to answer the following quenstions:
\begin{itemize}
    \item What is the optimal controller to make the identification accuracy within a certain value, and subsequently,
    \item what is the optimal key length needed to secure the closed-loop system within a life span of the system?
\end{itemize}
To this end, we introduce two novel security quantities, \textit{sample identifying complexity} and \textit{sample deciphering time}.
This type of quantification is not reported in any papers on cryptography.

The sample identifying complexity is derived as a lower bound for the total variance, i.e., the inverse of precision, of Bayesian estimation by an adversary.
The sample deciphering time is computation time for breaking encrypted data without a secret key to obtain a data set for the estimation.
The security in this study is defined based on these quantities.
Roughly speaking, we say an encrypted control system is secure if the adversary cannot identify the system matrix with a certain precision within a life span of the system.
The formal definition of the security will be described later.

The sample deciphering time is introduced in two cases with static-key encryption and dynamic-key encryption.
Static-key encryption is traditional public-key encryption of which the key pair is identical throughout the communication.
In contrast, a key pair in dynamic-key encryption~\cite{Teranishi20_5} is updated at a short time interval, e.g., a sampling period.
Although dynamic-key encryption would improve the security level of encrypted control systems, its security has not yet been proved.
We extend the dynamic-key encryption scheme in~\cite{Teranishi20_5} and provide a security proof of the extended scheme.

Using the security quantities, we formulate a design problem of optimal key length and controller.
The optimal controller is designed to maximize the sample identifying complexity.
In other words, the controller maximizes the difficulty of the system identification.
More interestingly, such a controller is provided as the standard stochastic cheap controller improving the stability degree of a closed-loop system.
This fact means, in controller design, there is no trade-off between the security level and the control performance.

After designing the optimal controller, we design the optimal key length to secure an encrypted control system.
The optimal key length is obtained as the minimum key length to make the sample deciphering time longer than the system's life span.
This key length is beneficial for reducing implementation costs of an encrypted control system while keeping the security level because the size of key length has a trade-off between ciphertext strength and computation costs of encryption and decryption algorithms.

\subsection{Outline}

\secref{sec:preliminaries} summarizes notations and a definition of homomorphic encryption.
The ElGamal encryption, an example of a multiplicative homomorphic encryption scheme, is also introduced.
\secref{sec:problem_setting} describes the attack scenario considered in this study.
We define the security of encrypted control systems and formulate a design problem of the optimal key length and controller.
\secref{sec:curves} proposes sample identifying complexity and sample deciphering time.
They are used to understand the relationships among a key length, controller, and the number of samples for system identification.
\secref{sec:optimal_design} provides the solution to the problem based on the security quantities.
Additionally, we show how the security quantities can be used for other design problems in encrypted control systems.
\secref{sec:simulation} demonstrates the validity of the proposed method by numerical simulations.
\secref{sec:conclusion} concludes this paper and presents some remarks on the results of this study.

\section{Preliminaries}\label{sec:preliminaries}

\subsection{Notation}

The sets of real numbers, integers, security parameters, public keys, secret keys, plaintexts, and ciphertexts are denoted by $\R$, $\Z$, $\S$, $\Kp$, $\Ks$, $\M$, and $\C$, respectively.
We define the sets of integers $\Z^{+}\coloneqq\{z\in\Z\mid 0\le z\}$ and $\Z_{n}\coloneqq\{z\in\Z\mid 0\le z<n\}$.
The set of $n$-dimensional real column-vectors is denoted by $\R^{n}$, and that of $m$-by-$n$ real-valued matrices is denoted by $\R^{m\times n}$.
The $i$th element of a vector $v\in\R^{n}$ is denoted by $v_{i}$, and the $\ell_{2}$ norm and the maximum norm of $v$ are denoted by $\|v\|$ and $\|v\|_{\infty}$, respectively.
The $i$th column vector and $(i,j)$ entry of a matrix $M\in\R^{m\times n}$ are denoted by $M_{i}$ and $M_{ij}$, respectively.
The max norm and column stack vector of $M$ are defined by $\|M\|_{\max}\coloneqq\max_{i,j}\{|M_{ij}|\}$ and $\vec(M)\coloneqq[M_{1}^{\top}\cdots M_{n}^{\top}]^{\top}$, respectively.
The cardinality of a set $\mathcal{A}$ is denoted by $|\cal{A}|$.
The Gaussian distribution with a mean $\mu$ and a variance-covariance matrix $\Sigma$ is denoted by $\N(\mu,\Sigma)$.
The probability density function of $\N(\mu,\Sigma)$ is denoted by $f(x;\mu,\Sigma)$.

\begin{definition}
    Let $\mathcal{A}$ be a finite set and $X$ be a random variable.
    If $\Pr(X=a)=1/|\mathcal{A}|$, $\forall a\in \mathcal{A}$, then we say $X$ follows the discrete uniform distribution over $\mathcal{A}$ and is denoted as $X\sim\U(\mathcal{A})$.
\end{definition}

\begin{definition}[negligible function~\cite{Katz15}]
    We say a function $\epsilon:\Z^{+}\setminus\{0\}\to\R$ is negligible if for every positive integer $c>0$ there exists $N\in\Z$ such that $|\epsilon(n)|<n^{-c}$ holds for all $n>N$.
\end{definition}

\subsection{Homomorphic encryption and its example}

This section describes the definition and example of homomorphic encryption to introduce the encrypted-control framework.
One can refer~\cite{Acar18} for the detailed survey of homomorphic encryption.

A public-key encryption scheme is a triplet $(\Gen,\Enc,\Dec)$, where $\Gen:\S\to\Kp\times\Ks:k\mapsto(\pk,\sk)$ is a key generation algorithm, $\Enc:\Kp\times\M\to\C:(\pk,m)\mapsto c$ is an encryption algorithm, $\Dec:\Ks\times\C\to\M:(\sk,c)\mapsto m$ is a decryption algorithm, $k$ is a security parameter, e.g., a key length, and $(\pk,\sk)=\Gen(k)$ is a pair of public key and secret key.
$\Enc$ and $\Dec$ perform elementwise for a vector and a matrix.
Public-key encryption schemes must satisfy $\Dec(\sk,\Enc(\pk,m))=m$ for all $m\in\M$ and $(\pk,\sk)$ generated by $\Gen$.

\begin{definition}\label{def:mhe}
    We say $(\Gen,\Enc,\Dec)$ is multiplicative homomorphic encryption if $\Dec(\sk,c\boxtimes c')=mm'$ for all $m,m'\in\M$ and $c,c'\in\C$ satisfying $\Enc(\pk,m)=c$ and $\Enc(\pk,m')=c'$, where $\boxtimes:\C\times\C\to\C$ is a binary operation over $\C$.
    Similarly, additive homomorphic encryption is defined with $\boxplus:\C\times\C\to\C$.
\end{definition}

An example of multiplicative homomorphic encryption includes the ElGamal encryption~\cite{ElGamal85}.
Its algorithms are $\Gen:k\mapsto(\pk,\sk)=((p,q,g,h),s)$, $\Enc:(\pk,m)\mapsto c=(g^{r}\bmod p,mh^{r}\bmod p)$, and $\Dec:(\sk,(c_{1},c_{2}))\mapsto{c_{1}}^{-s}c_{2}\bmod p$, where $q$ is a $k$ bit prime, $p=2q+1$ is a safe prime, $g$ is a generator of a cyclic group $\G\coloneqq\{g^{i}\bmod p\mid i\in\Z_{q}\}=\M\subset\Z_{p}\setminus\{0\}$ such that $g^{q}\bmod p=1$, $h=g^{s}\bmod p$, $\C=\G^{2}$, and $r,s\sim\U(\Z_{q})$.
Additionally, multiplicative homomorphism is $\Dec(\sk,\Enc(\pk,m)\ast\Enc(\pk,m')\bmod p)=mm'\bmod p$, where $\ast$ is the Hadamard product.

\section{Attack Scenario and Problem Setting}\label{sec:problem_setting}

Consider a plant described by the discrete-time stochastic linear system
\begin{equation}
    x_{t+1}=A_{p}x_{t}+B_{p}u_{t}+w_{t},
    \label{eq:plant}
\end{equation}
where $t\in\Z^{+}$ is a time step, $x\in\R^{n}$ is a state, $u\in\R^{m}$ is an input, and $w\in\R^{n}$ is an i.i.d. random variable following the Gaussian distribution $\N(\bfzero,L^{-1})$ with the zero vector $\bfzero$ and a precision matrix $L$.
Assume that $(A_{p},B_{p})$ is controllable, and the initial state is given by $x_{0}\sim\N(\bfzero,L^{-1})$.
A state-feedback controller
\begin{equation}
    u_{t}=Fx_{t},
    \label{eq:controller}
\end{equation}
which is installed on a computer over a network, e.g., cloud, is employed for stabilizing \eqref{eq:plant}, where a feedback gain $F$ is to be designed.
Note that output-feedback controllers can also be considered although we use a state-feedback controller for the simplicity of discussion.

The networked control system with \eqref{eq:plant} and \eqref{eq:controller} has risks of eavesdropping attacks because the plant and controller communicate with each other via network links.
This study considers encrypted control proposed in~\cite{Kogiso15} as a secure control framework against the attacks.
An encrypted control system includes an encrypter $\Enc$ and decrypter $\Dec$ in its feedback loop; see \figref{fig:scenario}.
Note that a sensor (encrypter) and an actuator (decrypter) in this study are assumed to be installed on a unified computer of plant side.
An encrypted controller of \eqref{eq:controller} with multiplicative homomorphic encryption of \defref{def:mhe} is defined as
\[
    (c_{F},c_{x_{t}})\mapsto c_{U_{t}}=
    \begin{bmatrix}
        c_{F_{11}}\boxtimes c_{x_{1,t}} & \cdots & c_{F_{1n}}\boxtimes c_{x_{n,t}} \\
        \vdots & \ddots & \vdots \\
        c_{F_{m1}}\boxtimes c_{x_{1,t}} & \cdots & c_{F_{mn}}\boxtimes c_{x_{n,t}}
    \end{bmatrix},
\]
where $c_{F}=\Enc(\pk,F)$, and $c_{x_{t}}=\Enc(\pk,x_{t})$.
An input is restored as
\[
    u_{t}=
    \begin{bmatrix}
        \sum_{j=1}^{n}\Dec(\sk,c_{U_{1j,t}}) \\
        \vdots \\
        \sum_{j=1}^{n}\Dec(\sk,c_{U_{mj,t}})
    \end{bmatrix},
\]
and it approximately equals to an input of \eqref{eq:controller} if quantization errors caused by the encryption are sufficiently small.
Thus, the dynamics of the encrypted control system is obtained as
\begin{equation}
    x_{t+1}=Ax_{t}+w_{t},\quad A\coloneqq A_{p}+B_{p}F.
    \label{eq:system}
\end{equation}
By using an encrypted-control framework, conventional controllers can be used while their gains and signals over network links are encrypted.

\begin{figure}[!t]
    \centering
    \includegraphics[scale=1]{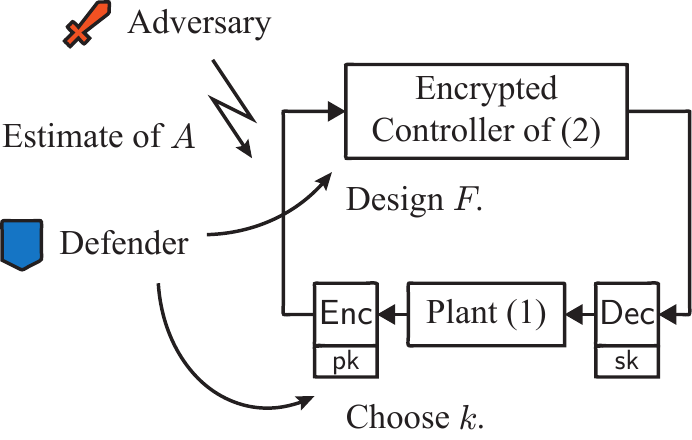}
    \caption{Attack scenario and actions of adversary and defender.}
    \label{fig:scenario}
\end{figure}

We consider an attack scenario to identify the dynamics of the encrypted control system.
The dynamics must be secret even though an adversary eavesdrops and deciphers the ciphertexts because he/she would exploit it as an initial step for executing more sophisticated attacks, such as stealth attacks.
As a result, the total security level of encrypted control systems can be improved by preventing the identification attack.
The worst scenario for a defender is Bayesian estimation of the dynamics, i.e., $A$ in \eqref{eq:system}, with deciphered data because the estimation is the best in terms of the variance of estimator.
This attack is formulated as follows:

\begin{definition}\label{def:adversary}
    An adversary follows the protocol below:
    \begin{enumerate}
        \item Given $T\in\Z^{+}$, collect $\Denc\coloneqq\{\Enc(\pk,x_{t})\}_{t=0}^{T}$ by eavesdropping attacks.
        \item Expose $\D\coloneqq\{x_{t}\}_{t=0}^{T}$ by breaking the ciphertexts in $\Denc$ using a computer of which performance is $\Upsilon$~floating-point operations per second (FLOPS).
        \item Choose a prior probability $p(A)=f(\vec(A);\mu,\Lambda^{-1})$ based on his/her knowledge about a target control system.
        Then, estimate a posterior probability $p(A|\D)=f(\vec(A);\hat{\mu}(T),\hat{\Lambda}^{-1}(T))$ by Bayesian estimation with $p(A)$ and $\D$.
    \end{enumerate}
\end{definition}

An adversary aims to identify a system matrix $A$ as a posterior probability $p(A|\D)$, and an estimation $\hat{A}$ is given by $\vec(\hat{A})=\hat{\mu}(T)$.

Is the encrypted control system secure under what conditions in these settings?
In this paper, the system is said to be secure if identification of $A$ with a certain precision is impossible within a given period.
In particular, the security in the attack scenario is defined as follows, where we use the fact that the trace of a variance-covariance matrix can be used for a measure of the precision of the estimation since it represents the total variance:

\begin{definition}\label{def:security}
    Let $\tau_{c}$ be a life span that represents a period until the system \eqref{eq:plant} is replaced, and $\gamma_{c}$ be an acceptable variance against adversary’s estimation.
    Define
    \[
        \tau(T,k)\coloneqq\text{time for executing step 2) in \defref{def:adversary}.}
    \]
    The encrypted control system in \figref{fig:scenario} is said to be \textit{secure} if there does not exist $T\in\Z^{+}$ satisfying
    \begin{equation}
        \EV{\tr(\hat{\Lambda}^{-1}(T))}<\gamma_{c}\land\tau(T,k)\le\tau_{c},
        \label{eq:security}
    \end{equation}
    where $\hat{\Lambda}(T)$ is defined in \defref{def:adversary}.
    If not, the system is said to be \textit{unsecure}.
\end{definition}

In \defref{def:security}, $\tau_{c}$ and $\gamma_{c}$ are the design parameters while the key length $k$ and the controller $F$ are the implicit decision variables.
As $\tau_{c}$ is taken larger for protecting the system during a longer period, the key length $k$ would be longer~\cite{Bernstein93}.
Although the longer $k$ is beneficial for ciphertext strength, it is not desirable in terms of implementation costs because the online computation costs of $\Enc$ and $\Dec$ with longer $k$ has to be larger~\cite{Chapter}.
In other words, the choice of a longer key length increases economic costs since a high performance computer is required for keeping the real-time operation of the control system.
Since there is such a trade-off, we will design $F$ for making the key length as short as possible.
Later we will show that the ease of identification relates to the stability of $A$ in \eqref{eq:system}.
This implies that the choice of a \textit{good} controller $F$ can make the key length $k$ shorter while making the precision of identification is within the tolerance $\gamma_{c}$.
In this light, we consider the following design problem for ensuring the dynamical system security.

\begin{problem}\label{prob:security}
    Consider the encrypted control system in \figref{fig:scenario} under the attack scenario in \defref{def:adversary}.
    Find $F$ and a minimum $k\in(0,\infty)$ such that the system is secure defined in \defref{def:security}.
\end{problem}

An essential question behind \probref{prob:security} is how the key length $k$, controller $F$, and the number of deciphered samples $T$ relate to the security.
The factors $k$ and $T$ are often taken into account in cryptography~\cite{Katz15} and sample complexity of computational learning theory~\cite{Kearns94}, respectively.
Unlike to this, we have to explicitly consider the controller gain as well as those two factors because the system of our interest has dynamics.
In view of this, \probref{prob:security} lies in between cryptography, learning theory, and control theory.
In the next section, we analyze the relation among $k$, $T$, $F$, and the security.

\begin{remark}
    Additive homomorphic encryption can also be used instead of using multiplicative homomorphic encryption.
    In such a case, an encrypted controller is defined as $(F,c_{x_{t}})\mapsto c_{u_{t}}=[F_{11}c_{x_{1,t}}\boxplus\cdots\boxplus F_{1n}c_{x_{n,t}},\ \cdots,\ F_{m1}c_{x_{1,t}}\boxplus\cdots\boxplus F_{mn}c_{x_{n,t}}]^{\top}$, and an input is given by $u_{t}=\Dec(\sk,c_{u_{t}})$.
    Note that the encrypted controller has an unencrypted parameter $F$, unlike one with multiplicative homomorphic encryption.
\end{remark}

\begin{remark}
    Although most algorithms to recover $\D$ from $\Denc$ would include integer operations rather than floating-point operations, the computational ability for integer operations in this study is assumed to be quantified by FLOPS. 
\end{remark}

\begin{remark}
    So far we have assumed that an adversary can exactly recover $\D$ from $\Denc$ without quantization errors caused by the encryption.
    In practice, $\Enc$ and $\Dec$ in \figref{fig:scenario} have to be equipped with an encoder $\Ecd:\R\to\M$ and decoder $\Dcd:\M\to\R$ that convert real numbers $(F,x_{t})$ to a plaintext space because the most existing homomorphic encryption schemes rely on arithmetic operations over integers.
    Therefore, quantization errors are always involved in the deciphered samples.
    However, for simplifying the following arguments, we do not consider the error, which is the \textit{worst} case scenario for the defender.
    The details of the quantization error analysis is described in \appref{app:quantization}. 
\end{remark}

\begin{remark}
    An adversary of this study is assumed to estimate a system matrix $A$ in \eqref{eq:system}.
    One may think that considering an estimation attack for a controller gain $F$ in \eqref{eq:controller} is also important.
    The attack can be treated as solving simultaneous equations for $F$ with $N$ independet data sets of $x$ and $u$.
    In such a case, the encrypted control system is said to be secure if $\tau(N,k)\le\tau_{c}$ is satisfied.
\end{remark}

\section{Sample Identifying-complexity Curve and Sample Deciphering-time Curve}\label{sec:curves}

This section introduces two novel quantities referred to as \textit{sample identifying-complexity curve} and \textit{sample deciphering-time curve} to clearly understand the relationship among $k$, $T$, $F$, and the security.

\subsection{Sample identifying-complexity curve}

We introduce the following lemma that connects the notion of the security in \defref{def:security} to the dynamics of \eqref{eq:system}. 

\begin{lemma}\label{lem:gamma}
    Consider the system in \figref{fig:scenario} under the attack in \defref{def:adversary}.
    Suppose $A$ in \eqref{eq:system} is Schur.
    Then, the parameters of a posterior probability $p(A|\D)$ in \defref{def:adversary} are described as
    \begin{align}
        \hat{\Lambda}(T)&=\Lambda+\sum_{t=0}^{T-1}(x_{t}\otimes I)L(x_{t}\otimes I)^{\top}, \label{eq:est_variance} \\
        \hat{\mu}(T)&=\hat{\Lambda}^{-1}(T)\left(\Lambda\mu+\sum_{t=0}^{T-1}(x_{t}\otimes I)Lx_{t+1}\right). \label{eq:est_mean}
    \end{align}
    Besides, the following relations hold:
    \begin{align}
        &\tr(\hat{\Lambda}^{-1}(T))\ge\cfrac{n^{2}}{\displaystyle\tr(\Lambda)+\tr(L)\sum_{t=0}^{T-1}\|x_{t}\|^{2}}, \label{eq:gamma_norm} \\
        &\EV{\tr(\hat{\Lambda}^{-1}(T))} \nonumber \\
        &\!\ge\!\gamma(T,F)\!\coloneqq\!\cfrac{n^{2}}{\displaystyle\tr(\Lambda)\!+\!\tr(L)\!\sum_{t=0}^{T-1}\tr\!\left(\sum_{i=0}^{t}A^{i}L^{-1}(A^{i})^{\top}\right)}.
        \label{eq:gamma}
    \end{align}
\end{lemma}

\begin{proof}
    See \appref{proof:lem_gamma}.
\end{proof}

Note here that the term $\sum_{i=0}^{t}A^{i}L^{-1}(A^{i})^{\top}$ is the weighted finite-time controllability gramian of \eqref{eq:system}.
\lemref{lem:gamma} shows that the quantification $\EV{\tr(\hat{\Lambda}^{-1}(T))}$ in \eqref{eq:security} can be bounded from below by using the trace of the gramian.
We refer to the bound $\gamma$ as \textit{sample identifying-complexity curve} due to the acknowledge that the curve captures the complexity of the identification of $\vec(A)$ with $|\D|=T+1$ samples.
It should be noted here that the system trajectory $x_{t}$ explicitly depends on the controller gain $F$.
Thus, the curve is a function of $T$ and $F$.
We can see the following two observations from \eqref{eq:gamma}.
\begin{itemize}
    \item Dependency of $F$: The sample identifying complexity is larger if $F$ makes the stability degree measured by the trace of the controllability gramian smaller.
          This is natural because as the system more stable, the amount of information, i.e., the system output driven by the initial state and external input $w$, can be less, thereby making the identification more difficult.
    \item Dependency of $T$: The sample identifying complexity is larger if the number of deciphered samples lesser.
          This implies that the identification is difficult for the adversary by decreasing leaked data samples.
\end{itemize}
\figref{fig:sic} depicts the schematic picture of the curve $\gamma(T,F)$.
Although $F$ is an $m$-by-$n$ matrix, in the figure larger $F$ implies the one making $A$ in \eqref{eq:system} more stabilized.

When a sample size $T$ is sufficiently large, the summation terms in \eqref{eq:est_variance} and \eqref{eq:est_mean} should be much larger than the terms $\Lambda$ and $\Lambda\mu$, respectively.
Then, the estimates $\hat{\Lambda}$ and $\hat{\mu}$ satisfy the following relation to an estimation error.

\begin{corollary}\label{cor:est_error}
    If a sample size $T$ is sufficiently large, then $\hat{\Lambda}(T)$ in \eqref{eq:est_variance} satisfies
    \begin{equation}
        \EV{\|A-\hat{A}\|_{F}^{2}}=\EV{\tr(\hat{\Lambda}^{-1}(T))},
        \label{eq:est_error}
    \end{equation}
    where $\vec(\hat{A})=\hat{\mu}(T)$, and $\hat{\mu}(T)$ is given by \eqref{eq:est_mean}.
\end{corollary}

\begin{proof}
    See \appref{proof:cor_est_error}.
\end{proof}

The equalities \eqref{eq:gamma} and \eqref{eq:est_error} show that the sample identifying complexity is a lower bound of a type of estimation error with a sufficiently large samples.
Thus, it is reasonable that the difficulty of identifying a system matrix is quantified by the sample identifying complexity.

\begin{figure*}[!t]
    \centering
    \subfigure[Sample identifying-complexity curve.]{\includegraphics[scale=1]{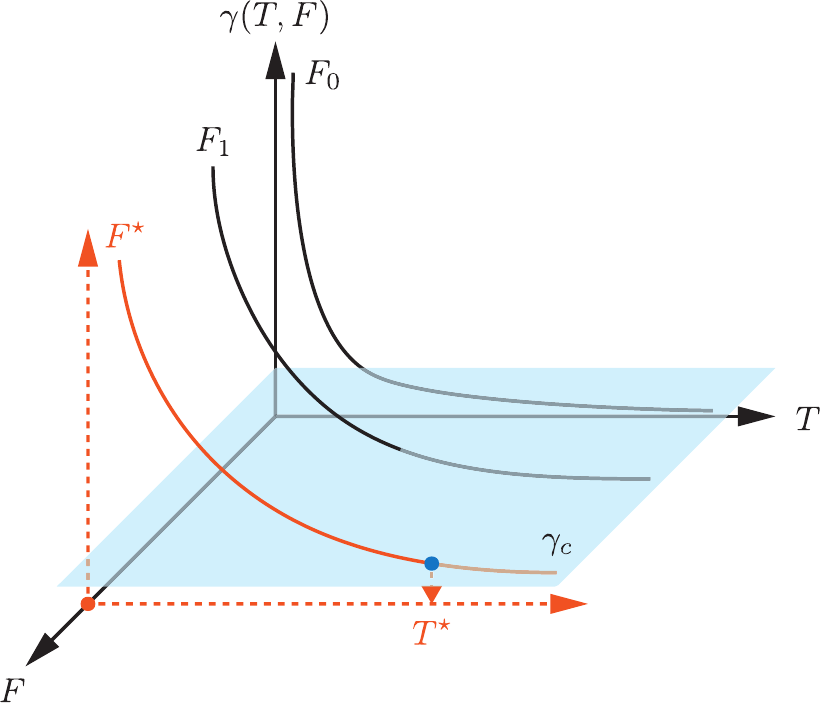}\label{fig:sic}}%
    \subfigure[Sample deciphering-time curve.]{\includegraphics[scale=1]{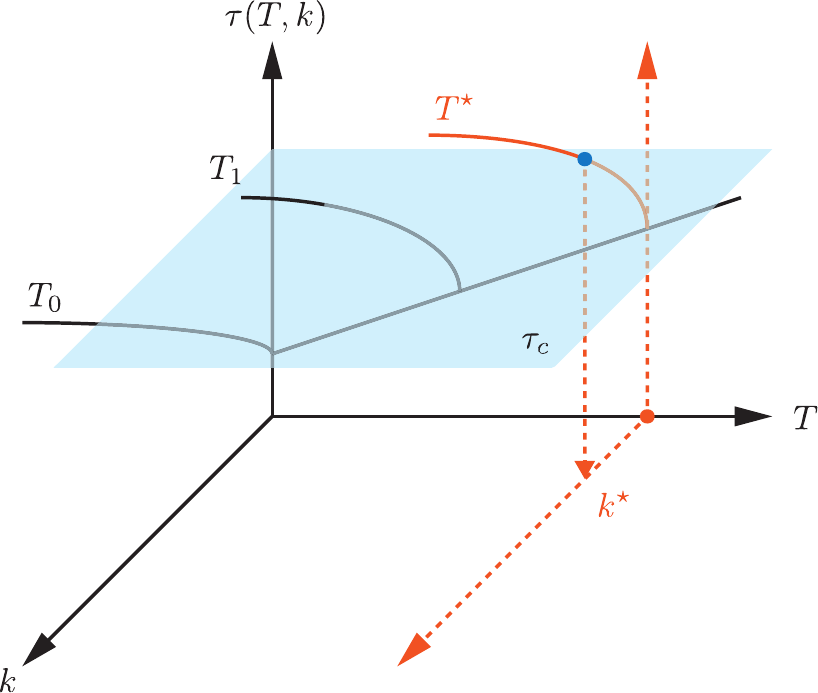}\label{fig:sdt}}
    \caption{Schematic pictures of sample identifying-complexity curve $\gamma(T,F)$ and sample deciphering-time curve $\tau(T,k)$.}
    \label{fig:curve}
\end{figure*}

For the following argument, we show a special case when $L^{-1}=\sigma^{2}I$ and the adversary has no prior information about the system, i.e.,
\begin{equation}
    \tr(\Lambda)=0.
    \label{eq:no_info}
\end{equation}
Then, the following corollary immediately follows from \lemref{lem:gamma}:

\begin{corollary}\label{cor:gamma}
    If $L^{-1}=\sigma^{2}I$ and \eqref{eq:no_info} hold, then $\gamma$ in \eqref{eq:gamma} satisfies
    \begin{equation}
        \gamma(T,F)=\cfrac{n}{\displaystyle\sum_{t=0}^{T-1}\tr\left(\sum_{i=0}^{t}A^{i}(A^{i})^{\top}\right)}.
        \label{eq:gamma_simple}
    \end{equation}
\end{corollary}

\begin{proof}
    See \appref{proof:cor_gamma}.
\end{proof}

The sample identifying-complexity curve connects the relationship between the sample complexity $\EV{\tr(\hat{\Lambda}^{-1}(T))}$ in \defref{def:security} and a pair $(T,F)$.
Before showing how this is useful for solving \probref{prob:security}, we next show a different curve that connects the security to $T$ and a key length $k$.

\begin{remark}
    We have introduced the expectation of a lower bound of $\hat{\Lambda}^{-1}$ because the computation of the inverse of $\hat{\Lambda}$, in general, requires a large number of computation resources, and it cannot be computed in advance of the control system's operation.
    A similar approach can be found in~\cite{Farokhi18,Farokhi19,Farokhi20_2,Ziemann20}, and the studies employed the inverse of a trace of the Fisher information matrix as a lower bound of the precision of general unbiased estimator for dynamical systems.
    Unfortunately, the approach is not specialized in our attack scenario, i.e., it would give a loose lower bound of $\tr(\hat{\Lambda}^{-1})$, and the lower bound cannot be computed without the system's operating data.
\end{remark}

\subsection{Sample deciphering-time curve}

In this paper, we refer to $\tau(T,k)$ in \eqref{eq:security} as \textit{sample deciphering-time curve} due to the acknowledge that the curve captures the computation time for deciphering the $T+1$ ciphertexts of $\Denc$.
One might consider that the deciphering time does not depend on the number of samples.
This is true in a traditional setup of public-key encryption, referred to as \textit{static-key encryption} in this paper, where the keys used for encrypting all samples are identical~\cite{Katz15}.
On the other hand, when the keys of individual samples are completely different, in other words, \textit{dynamic-key encryption} is used~\cite{Teranishi20_5}, the deciphering-time clearly depends on the number of samples.
We show an explicit representation of $\tau(T,k)$ for each encryption scheme, and show an advantage of the dynamic case in terms of the security in \defref{def:security}.

\subsubsection{Static-key case}

As a multiplicative homomorphic encryption scheme, this study uses the ElGamal encryption $\E$ described in \secref{sec:preliminaries}.
The security of $\E$, i.e., the difficulty of breaking the encryption, is based on the hardness of the discrete logarithm problem for $\G$ that is defined as follows:

\begin{definition}[discrete logarithm problem~\cite{Hoffstein08}]
    Let $G$ be a group with a binary operation~$\circ$.
    The discrete logarithm problem (DLP) for $G$ is to determine, for any given elements $g,h\in G$, an integer $x$ satisfying
    \[
        g^{x}=\underbrace{g\circ g\circ \cdots \circ g}_{x\text{ times}}=h.
    \]
    Additionally, the assumption that there does not exist a polynomial-time algorithm to solve the DLP is called the discrete logarithm assumption.
\end{definition}

In the field of cryptography, the discrete logarithm assumption is widely believed to be satisfied.
The ElGamal encryption achieves indistinguishability against chosen-plaintext attacks (IND-CPA) under the decisional Diffie-Hellman (DDH) assumption~\cite{Katz15} that is a variant of the discrete logarithm assumption.
The security level of IND-CPA means that an adversary can obtain no information about plaintexts from ciphertexts.
Hence, an adversary must solve the DLP for $\G$ to obtain $\D$ from $\Denc$.
The majority of algorithms for solving the DLP for a finite field $\F_{\eta}$ with a modulus $\eta$ are subexponential-time algorithms of which computation time is described as
\begin{equation}
    L_{v,d}(\eta)=\exp{\{d(\ln{\eta})^{v}(\ln{\ln{\eta}})^{1-v}\}},
    \label{eq:subexp_time}
\end{equation}
where $v$ and $d$ are algorithm parameters~\cite{Hoffstein08}.
For instance, the general number field sieve, the known fastest classical subexponential-time algorithm, has $v=1/3$ and $d=(64/9)^{1/3}$ in \eqref{eq:subexp_time}~\cite{Bernstein93}.
Thus, we use
\begin{equation}
    L(k)\coloneqq L_{1/3,(64/9)^{1/3}}(2^{k})
    \label{eq:subexp_time_key}
\end{equation}
as the computation time for deciphering a ciphertext of $\E$ with a key length $k$ in the following.
Note that $\G\subset\F_{p}$, and $L(k)$ satisfies $L(k)\le L_{1/3,(64/9)^{1/3}}(p)$ since $p\in(2^{k},2^{k+1})$.
Therefore, \eqref{eq:subexp_time_key} is stricter with a defender than \eqref{eq:subexp_time}.

We next show the sample deciphering time of the static-key encryption.
Since a single key pair is used for encrypting all the samples throughout a life span of the encrypted control system, the adversary has to break only one ciphertext for finding the secret key.
Once the secret key is found, he/she can decrypt all ciphertexts of $\Denc$ immediately.
Thus, the sample deciphering time in this case can be described as
\begin{equation}
    \tau(0,k)=\cfrac{L(k)}{\Upsilon},
    \label{eq:tau_0}
\end{equation}
where $\Upsilon$ is defined in \defref{def:adversary}.
For satisfying the second inequality of \eqref{eq:security}, the key length $k$ will be long because even only one ciphertext cannot be broken during a given period $\tau_{c}$.
Although it is natural from the ordinary manner in cryptography, the online computation costs of the associated $\Enc$ and $\Dec$ in \figref{fig:scenario} must be heavy, which is not desirable for real-time controls.

\begin{remark}
    The number field sieve is used for solving not only the DLP but also the prime factorization problem.
    Thus, \eqref{eq:tau_0} also enables to estimate the computation times for breaking other encryption schemes, such as RSA~\cite{Rivest78} and Paillier encryption~\cite{Paillier99}.
    Moreover, we can change $L(k)$ according to any given encryption scheme such as the LLL-algorithm for lattice and fully homomorphic encryption, and, therefore, the sample deciphering time can be obtained for any encryption scheme as well.
\end{remark}

\subsubsection{Dynamic-key case}

One way to reduce the online computational costs of $\Enc$ and $\Dec$ while keeping the sample deciphering time long is to regenerate a secret key at each sampling time.
However, this approach is not suitable for real-time controls due to the high computational costs.
As an alternative approach, we employ the dynamic-key encryption~\cite{Teranishi20_5} that is an augmented concept of public-key encryption.
The overview is as follows:
First, give a key pair by $\Gen$ at the initial time.
The secret key at time $t+1$ is computed by a simple updating rule based on a modulus operation with a random number and the secret key at time $t$.
At the same time, a public key and ciphertexts of controller parameters are also updated to keep the correctness, i.e., the property that a ciphertext is decrypted correctly, with the new secret key.
Due to the time-dependency of this dynamic-key encryption, the adversary would have to break $T+1$ ciphertexts to collect $\D$ from $\Denc$.
However, the security proof of the dynamic-key encryption has not yet been shown.
Additionally, the dynamic-key encryption refreshes only the second element of ciphertext, and so, the first element remains the same value.
In the following, we extend the dynamic-key encryption in~\cite{Teranishi20_5} to update all components of ciphertext and provide the security proof of the scheme.

The dynamic ElGamal encryption in this study is constructed as follows:

\begin{definition}\label{def:dynenc}
    Dynamic ElGamal encryption is a tuple $\Edyn\coloneqq(\Gen,\Enc,\Dec,T_{\K},T_{\C})$ with the transition maps
    \begin{align*}
        T_{\K}&:((p,q,g,h),s)\mapsto((p,q,g,hg^{s'}\bmod p),s\!+\!s'\bmod q), \\
        T_{\C}&:(c_{1},c_{2})\mapsto(c_{1}g^{r'}\bmod p,(c_{1}g^{r'})^{s'}c_{2}h^{r'}\bmod p),
    \end{align*}
    where $r',s'\sim\U(\Z_{q})$.
\end{definition}

\begin{remark}
    The random number $s'$ needs to be shared secretly between a sensor an an actuator if they are installed on different places.
    This can be achieved by using a standard symmetric-key encryption scheme, such as AES.
    Similarly, a plant can transmit $r'$ and $s'$ to an encrypted controller secretly.
\end{remark}

In \defref{def:dynenc}, $T_{\K}$ and $T_{\C}$ imply updating rules for a key pair and ciphertext, respectively.
$T_{\C}$ of the dynamic ElGamal encryption updates both $c_{1}$ and $c_{2}$ unlike to the scheme in~\cite{Teranishi20_5}.
We first show that the correctness and multiplicative homomorphism of our encryption scheme are satisfied even though the transition map is modified.

\begin{proposition}\label{prop:correctness}
    Let $k$ be a key length, $(\pk_{0},\sk_{0})=\Gen(k)$, and $c_{0}=\Enc(\pk_{0},m)$.
    If $(\pk_{t+1},\sk_{t+1})=T_{\K}(\pk_{t},\sk_{t})$ and $c_{t+1}=T_{\C}(c_{t})$, then
    \[
        \Dec(\sk_{t},c_{t})=\Dec(\sk_{t},\Enc(\pk_{t},m))=m\bmod p
    \]
    for all $m\in\M$ and $t\in\Z^{+}$.
    Furthermore, the multiplicative homomorphism
    \[
        \Dec(\sk_{t},c_{t}\ast\Enc(\pk_{t},m')\bmod p)=mm'\bmod p
    \]
    is satisfied for all $m,m'\in\M$ and $t\in\Z^{+}$.
\end{proposition}

\begin{proof}
    See \appref{proof:correctness}.
\end{proof}

Due to the homomorphism, the dynamics of the encrypted control system in \figref{fig:scenario} with the dynamic-key encryption scheme $\Edyn$ can be regarded as \eqref{eq:system} while the key pair and ciphertexts are dynamically updated. 

We next show an explicit representation of the sample deciphering-time curve $\tau(T,k)$ when $\Edyn$ is used.
To this end, we show a cryptographic property of the transition maps $T_{\K}$ and $T_{\C}$.

\begin{proposition}\label{prop:negligible}
    Let $k$ be a key length, $(\pk_{0},\sk_{0})=\Gen(k)$, $m\in\M$, and $c_{0}=\Enc(\pk_{0},m)$.
    A key pair and ciphertext are updated by $(\pk_{t+1},\sk_{t+1})=T_{\K}(\pk_{t},\sk_{t})$ and $c_{t+1}=T_{\C}(c_{t})$, respectively.
    Suppose an adversary knows $\pk_{t}$, $\sk_{t}$, and $c_{t}$ and can solve the DLP for $\G$.
    There exists a negligible function $\epsilon(k)$ such that
    \[
        \Pr(\hat{\sk}_{t+1}=\sk_{t+1})<\epsilon(k)\land\Pr(\hat{\sk}_{t-1}=\sk_{t-1})<\epsilon(k),
    \]
    for all $t\ge1$, where $\hat{\sk}_{t+1}$ and $\hat{\sk}_{t-1}$ are adversary's estimations of $\sk_{t+1}$ and $\sk_{t-1}$, respectively.
\end{proposition}

\begin{proof}
    See \appref{proof:negligible}.
\end{proof}

\propref{prop:negligible} implies that if we use the dynamic ElGamal cryptosystem, probability that an adversary can obtain the secret keys at time $t+1$ and $t-1$ is negligibly small even though he/she knows all information at time $t$ including the information given by solving the DLP for $\G$ as long as the updates of a key pair and ciphertexts are performed secretly.
This fact derives the following proposition on the security of our encryption scheme.

\begin{proposition}\label{prop:IND-CPA}
    $\Edyn$ satisfies IND-CPA at time $t$ under the DDH assumption even though an adversary knows $\{\pk_{i}\}_{i=0}^{t-1}$ and $\{\sk_{i}\}_{i=0}^{t-1}$.
\end{proposition}

\begin{proof}
    See \appref{proof:IND-CPA}.
\end{proof}

From \propsref{prop:negligible}{prop:IND-CPA}, an adversary cannot obtain any information about a secret key and plaintext for all time $t$ even though he/she has secret keys at time $t-1$ and $t+1$.
Thus, he/she must solve the DLP for $\G$ $T+1$ times to collect $\D$ from $\Denc$.
Therefore, the computation time for deciphering ciphertexts of $\Denc$ is linearly increased from $\tau(0,k)$ as a sample size of $\Denc$ increases if $\Edyn$ is used.
Thus, the following lemma is derived:

\begin{lemma}\label{lem:tau}
    A sample deciphering-time curve of the encrypted control system in \figref{fig:scenario} with $\Edyn$ in \defref{def:dynenc} is given as
    \begin{equation}
        \tau(T,k)=\cfrac{(T+1)L(k)}{\Upsilon},
        \label{eq:tau}
    \end{equation}
    where $L(k)$ and $\Upsilon$ are defined in \eqref{eq:subexp_time_key} and \defref{def:adversary}, respectively.
\end{lemma}

Notice that the sample deciphering time of the static-key case corresponds to \eqref{eq:tau} with $T=0$. 
It should be noted here that the curve $\tau(T,k)$ monotonically increases as either of $T$ and $k$ increases.
A schematic picture of the sample deciphering-time curve is shown in \figref{fig:sdt}. 

In conclusion, for the encrypted control system in \figref{fig:scenario} with $\Edyn$ in \defref{def:dynenc}, we have introduced the two curves: 
\begin{itemize}
    \item $\gamma(T,F)$ in \eqref{eq:gamma} that characterizes the difficulty of identifying the system \eqref{eq:system}, and
    \item $\tau(T,k)$ in \eqref{eq:tau} that quantifies the difficulty of deciphering encrypted samples.
\end{itemize}
In the next section, we show how these two curves are useful for solving \probref{prob:security}.

\section{Optimal Key Length and Controller Design}\label{sec:optimal_design}

For simplifying the following discussion, we suppose that the assumptions in \corref{cor:gamma} hold.
From \defref{def:security} and \corref{cor:gamma}, the following immediately follows: Given $k$ and $F$, if there does not exist $T$ satisfying
\begin{equation}
    \gamma(T,F)<\gamma_{c}\land\tau(T,k)\le\tau_{c},
    \label{eq:opt_security}
\end{equation}
where $\gamma$ and $\tau$ are respectively in \eqref{eq:gamma_simple} and \eqref{eq:tau}, then the encrypted control system in \figref{fig:scenario} with $\Edyn$ in \defref{def:dynenc} and $(k,F)$ is secure.
An idea for designing a key length and controller based on the sample identifying-complexity curve $\gamma$ and sample deciphering-time curve $\tau$ is as follows:

\begin{itemize}
    \item Controller design: Note from \eqref{eq:gamma_simple} that the identification variance monotonically decreases as the number of samples increases because the finite-time controllability gramian
          \begin{equation}
              W_{t}\coloneqq\sum_{i=0}^{t}A^{i}(A^{i})^{\top}
              \label{eq:gramian}
          \end{equation}
          is positive definite.
          Thus, we should design the controller $F^{\star}$ that maximizes the minimum time step $T^{\star}$ satisfying $\gamma(T^{\star},F^{\star})<\gamma_{c}$.
    \item Key length design: The computation time for deciphering $|\Denc|=T^{\star}+1$ ciphertexts is $\tau(T^{\star},k)$, and the time monotonically increases in a key length $k$.
          Considering that a key length is desirable to be as small as possible from the perspective of computational costs, it should be designed as the minimum key length $k^{\star}$ satisfying $\tau(T^{\star},k^{\star})>\tau_{c}$.
\end{itemize}

The pair $(k^{\star},F^{\star})$ is a solution to \probref{prob:security} since there does not exists $T$ satisfying \eqref{eq:opt_security} with $(k^{\star},F^{\star})$.
Note here that the controller $F^{\star}$ simultaneously minimizes the trace of $W_{t}$ in \eqref{eq:gramian} of \eqref{eq:system}.
Hence, the controller also improves the stability of the control system, which will be discussed later.
In the following, the concrete design processes of $k^{\star}$ and $F^{\star}$ are described.

\subsection{Controller design}

Following the controller design step, we design $F^{\star}$ so that the minimum time step $T^{\star}$ satisfying the first inequality of \eqref{eq:opt_security} is as large as possible.
From \eqref{eq:gamma_norm}, this design can be solved by making the cost function
\begin{equation}
    J_{T}\coloneqq\EV{\sum_{t=0}^{T-1}\|x_{t}\|^{2}}
    \label{eq:cost}
\end{equation}
as small as possible.
Since this is a finite-horizon stochastic linear quadratic regulator (s-LQR) design problem, an optimal solution is given as follows:

\begin{lemma}\label{lem:opt_input}
    Consider the system \eqref{eq:system} and $J_{T}$ in \eqref{eq:cost}.
    Assume that $B_{p}$ is full column rank.
    Then, the control sequence
    \begin{equation}
        u_{t}=-(B_{p}^{\top}P_{t+1}B_{p})^{-1}B_{p}^{\top}P_{t+1}A_{p}x_{t},\quad t\in[0,T)
        \label{eq:opt_input}
    \end{equation}
    minimizes $J_{T}$, where
    \[
        P_{t}\!=\!A_{p}^{\top}P_{t+1}A_{p}\!-\!A_{p}^{\top}P_{t+1}B_{p}(B_{p}^{\top}P_{t+1}B_{p})^{-1}B_{p}^{\top}P_{t+1}A_{p}\!+\!I,
    \]
    and $P_{T}=I$.
\end{lemma}

\begin{proof}
    See \appref{proof:opt_input}.
\end{proof}

Although the control \eqref{eq:opt_input} is optimal, the resultant controller has to be time-varying.
Unfortunately, time-varying controllers are difficult to be used in the encrypted-control framework because controller parameters must be encrypted and stored in advance before controller operation due to the difficulty of encrypted controller parameters update.
On the other hand, as $T\rightarrow\infty$, the control law converges to $u_{t}=F^{\star}x_{t}$ with
\begin{equation}
    F^{\star}=-(B_{p}^{\top}PB_{p})^{-1}B_{p}^{\top}PA_{p},
    \label{eq:F_star}
\end{equation}
where $P>0$ is the solution to the discrete-time algebraic Riccati equation
\[
    P=A_{p}^{\top}PA_{p}-A_{p}^{\top}PB_{p}(B_{p}^{\top}PB_{p})^{-1}B_{p}^{\top}PA_{p}+I.
\]
Hence, as a suboptimal solution to make $J_{T}$ as small as possible, we use the static feedback gain $F^{\star}$ in \eqref{eq:controller}.
It is interesting that the standard stochastic cheap control \eqref{eq:opt_input} is a good solution from the perspective of the security.
This fact clearly connects the notion of the security and classical control theory.
Moreover, the fact means no trade-off between the security and the control performance exists in controller design under the adversary of \defref{def:adversary}.
In other words, whenever the defender wants $F^{\star}$ in \eqref{eq:F_star} for improving closed-loop damping performance, the controller is also a good solution in terms of the security.

Once $F^{\star}$ is designed, the minimum time step $T^{\star}$ satisfying the first inequality of \eqref{eq:opt_security} can be uniquely determined as follows:
\begin{align}
    &T^{\star}=\argmin_{T}E(T),\quad E(T)\coloneqq\sum_{t=0}^{T-1}\tr(W_{t}) \label{eq:T_star} \\
    &\text{s.t.}\quad E(T)>\cfrac{n}{\gamma_{c}}, \nonumber
\end{align}
where $W_{t}$ is defined in \eqref{eq:gramian}.
An illustrative interpretation of this optimization is shown by the red line in \figref{fig:sic}.
It should be noted here that $T^{\star}$ in \eqref{eq:T_star} can be determined for any controller as long as $A$ in \eqref{eq:system} is Schur.
However, $T^{\star}$ in this case will be larger than the one when $F^{\star}$ in \eqref{eq:F_star} is used.
This choice, as we will show later, induces a longer key length.
For tractable computation of $E(T)$, we introduce the following proposition.

\begin{proposition}\label{prop:sum_tr_Wt}
    The summation of trace of a finite-time controllability gramian $E(T)$ in \eqref{eq:T_star} can be computatd as
    \[
        E(T)\!=\!\left\{
        \begin{alignedat}{2}
            &n, &\ &T=1, \\
            &2n+\tr(AA^{\top}), &\ &T=2, \\
            &2E(T\!-\!1)\!-\!E(T\!-\!2)\!+\!\tr(A^{T-1}(A^{T-1})^{\top}), &\ &T\ge3.
        \end{alignedat}
        \right.
    \]
\end{proposition}

\begin{proof}
    See \appref{proof:sum_tr_Wt}
\end{proof}

Although the computational complexity for computing $E(T)$ by the definition is more than $O(T^{2})$, that by \eqref{eq:T_star} can be reduced to $O(T)$, which facilitates the optimization problem \eqref{eq:T_star}.
The obtained $T^{\star}$ is used for designing the minimum key length design problem, which is described in the next section.

\subsection{Key length design}

Suppose that $T^{\star}$ is given by \eqref{eq:T_star}.
Following the key length design step, we find a minimum key length $k^{\star}$ such that the second inequality of \eqref{eq:opt_security} does not hold.
It follows from the second inequality of \eqref{eq:opt_security} and $\tau(T,k)$ in \eqref{eq:tau} that the key length minimization can be summarized as
\begin{equation}
    k^{\star}=\argmin_{k}L(k)\quad\text{s.t.}\quad L(k)>\cfrac{\tau_{c}\Upsilon}{T^{\star}+1}.
    \label{eq:k_star}
\end{equation}
An illustrative interpretation of this optimization is shown by the red line in \figref{fig:sdt}.

In conclusion, we have the following the theorem.

\begin{theorem}\label{thm:design}
    Consider \probref{prob:security} with the assumptions in \corref{cor:gamma}.
    The controller $F^{\star}$ and the minimum key length $k^{\star}$ are given by \eqref{eq:F_star} and \eqref{eq:k_star}, respectively.
    Then, the encrypted control system in \figref{fig:scenario} with the dynamic-key encryption scheme $\Edyn$ in \defref{def:dynenc} is secure.
\end{theorem}

\begin{proof}
    See \appref{proof:design}.
\end{proof}

A pseudocode of the design algorithm is summarized as \argref{alg:opt_ecs}.

\begin{figure}[!t]
    \begin{algorithm}[H]
        \caption{Optimal design of encrypted control system with dynamic-key encryption}
        \label{alg:opt_ecs}
        \begin{algorithmic}
            \Require $A_{p}$, $B_{p}$, $n$, $\gamma_{c}$, $\tau_{c}$, and $\Upsilon$.
            \Ensure $F^{\star}$ and $k^{\star}$.
            \State \# Controller design.
            \State Solve $P=A_{p}^{\top}PA_{p}-A_{p}^{\top}PB_{p}(B_{p}^{\top}PB_{p})^{-1}B_{p}^{\top}PA_{p}$.
            \State $F^{\star}\gets -(B_{p}^{\top}PB_{p})^{-1}B_{p}^{\top}PA_{p}$.
            \State \# Solve optimization problem of \eqref{eq:T_star}.
            \State $A\gets A_{p}+B_{p}F^{\star}$.
            \State $x\gets n$, $y\gets 0$, $z\gets 0$, $T^{\star}\gets 1$.
            \While{$x\le n/\gamma_{c}$}
                \State $T^{\star}\gets T^{\star}+1$.
                \State $z\gets y$.
                \State $y\gets x$.
                \State $x\gets 2y-z+\tr(A^{T^{\star}-1}(A^{T^{\star}-1})^{\top})$.
            \EndWhile
            \State \# Solve optimization problem of \eqref{eq:k_star}.
            \State $k^{\star}\gets 1$.
            \While{$L(k^{\star})\le\tau_{c}\Upsilon/(T^{\star}+1)$}
                \State $k^{\star}\gets k^{\star}+1$
            \EndWhile
            \State \Return $F^{\star}$, $k^{\star}$.
        \end{algorithmic}
    \end{algorithm}
    \vspace{-8mm}
\end{figure}

\subsection{Other design problems}

The parameters of a sample identifying-complexity curve $\gamma(T,F)$ and a sample deciphering-time curve $\tau(T,k)$ are a time step $T$, controller $F$, and key length $k$.
The optimal key length $k^{\star}$ in \probref{prob:security} is derived under a given controller $F^{\star}$.
Similarly, by fixing $T$ or $k$, the curves can be used for formulation of other design problems.

For example, a problem to design a controller gain $F$ under a given key length $k$ is a reverse problem of \probref{prob:security}.
A degree of freedom in design of $F$ in this problem is restricted by $k$ through the minimum time step $T^{\star}$ satisfying $\tau(T^{\star},k)>\tau_{c}$.
That is, a defender wants to find $F$ achieving a certain degree of stability of a control system, which is implicitly parameterized by $k$.

Furthermore, a problem to design $F$ and $k$ under the given time step $T=\tau_{c}/T_{s}$ is a variant of \probref{prob:security}, where $T_{s}$ is a sampling time.
An adversary in the variant is weaker than one in \probref{prob:security} because he/she uses all data within the life span for the estimation.
Thus, a defender would be required to design a finite-horizon controller maximizing $\gamma(T,F)$ and smaller key length than the solution to \probref{prob:security}.

\section{Numerical Simulation}\label{sec:simulation}

Consider \eqref{eq:plant} with
\[
    A_{p}=
    \begin{bmatrix}
        1 & 0.5 \\
        0 & -1.2
    \end{bmatrix}\!,\quad
    B_{p}=
    \begin{bmatrix}
        0 \\
        1
    \end{bmatrix}\!,\quad
    L=
    \begin{bmatrix}
        10^{4} & 0 \\
        0 & 10^{4}
    \end{bmatrix}\!.
\]
Let a controller $F$ in \eqref{eq:controller} be given so that the poles of $A$ in \eqref{eq:system} are assigned to $\pm0.99$.
We first show how the Bayesian estimation in step 3) of \defref{def:adversary} performs.
Let $\mu=\bfzero$ and $\Lambda=I$.
For each $T\in\{1,\ \dots,\ 5000\}$, we perform the estimation by using a data set $\D$.
\figref{fig:be} shows the result, where the blue lines are the estimated mean values ($\hat{A}_{11},\ \dots,\ \hat{A}_{22}$), and light-blue areas are the $95$~\% confidence intervals determined by $\hat{\Lambda}$.
The true values of $A$ are denoted by the dashed lines.
We can see from these figures that the precision of adversary's estimation improves as the number of samples increases.

\figref{fig:sicc} depicts the sample identifying-complexity curves $\gamma(T,F)$ in \eqref{eq:gamma} for different choices of the gains $F$ that assigns the poles of $A$ to $\pm0.99$, and $F^{\star}$ in \eqref{eq:F_star}.
Let the acceptable variance in \eqref{eq:opt_security} be chosen as $\gamma_{c}=10^{-6}$, which is denoted by the dashed line in the figure.
Then, the minimal time step $T^{\star}$ satisfying $\gamma(T^{\star},F)<\gamma_{c}$ is $18586$ while that for $F^{\star}$ is $384473$.
The time steps are denoted by $T_{F}^{\star}$ and $T_{F^{\star}}^{\star}$, respectively.
This result shows that the stochastic cheap controller \eqref{eq:F_star} improves sample identifying complexity of the closed-loop system.

We next compute the sample deciphering-time curves $\tau(T_{F}^{\star},k)$ and $\tau(T_{F^{\star}}^{\star},k)$ in \eqref{eq:tau}, and $\tau(0,k)$ in \eqref{eq:tau_0} for a comparison purpose.
Note here that the first (resp. second) represents the time for deciphering $T_{F}^{\star}+1$ (resp. $T_{F^{\star}}^{\star}+1$) ciphertexts of the dynamic ElGamal encryption $\Edyn$ in \defref{def:dynenc} while the third represents that for deciphering any ciphertext of the normal ElGamal encryption $\E$.
Note that the third case is irrelevant to controllers because the encryption is static-key encryption.
\figref{fig:sdtc} illustrates those three curves for $k$.
Let a life span and a computer performance be chosen as $\tau_{c}=1.5768\times10^{9}$~s ($50$~years), which is denoted by the dashed line in the figure, and $\Upsilon=442\times10^{15}$~FLOPS, which is the performance of Fugaku supercomputer~\footnote{https://www.top500.org/lists/top500/2020/11/}.
Then, by solving \eqref{eq:k_star}, the optimal key length for each cases is determined to $641$~bit, $734$~bit, and $1091$~bit.
This result implies that the simultaneous use of the dynamic-key encryption and the s-LQR optimal controller can drastically reduce the key length while keeping the security level of the encrypted control system, thereby reducing its computation costs.

Finally, we show how the differences of those three key lengths appear in the online computation times.
All the computations are done by using MacBook Pro (macOS Big Sur, $2.3$~GHz quad-core Intel Core i7, $32$~GB $3733$~MHz LPDDR4X) with C++.
The results are shown in \tabref{tab:ct}.
\figref{fig:ct} depicts the average computation times of $\Enc$, $\Dec$, and $T_{\K}$, which are performed on a plant side.
Their total times in $\E$, $\Edyn$ with $F$, and $\Edyn$ with $F^{\star}$ were $2.24$~ms, $0.94$~ms, and $0.61$~ms, respectively.
This result confirms that the computation time is decreased according to reducing the optimal key length by using the dynamic-key encryption and the optimal controller.
Although one may think that the resultant differences are not significant, the difference will be more significant for larger-dimensional systems.
This is because an online computation of encrypted control systems includes $n$ times of $\Enc$ and $mn$ times of $\Dec$ on a plant side.
Hence, for larger-dimensional systems, the proposed design methodology would be helpful for real-time controls while keeping the security level theoretically.

\begin{figure}[!t]
    \centering
    \subfigure[$\hat{A}_{11}=\hat{\mu}_{1}$.]{\includegraphics[scale=.5]{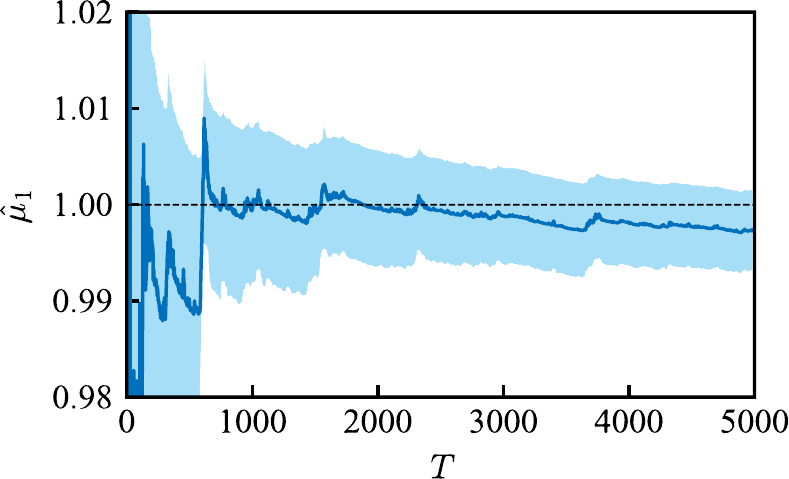}\label{fig:be_A11}}%
    \subfigure[$\hat{A}_{12}=\hat{\mu}_{3}$.]{\includegraphics[scale=.5]{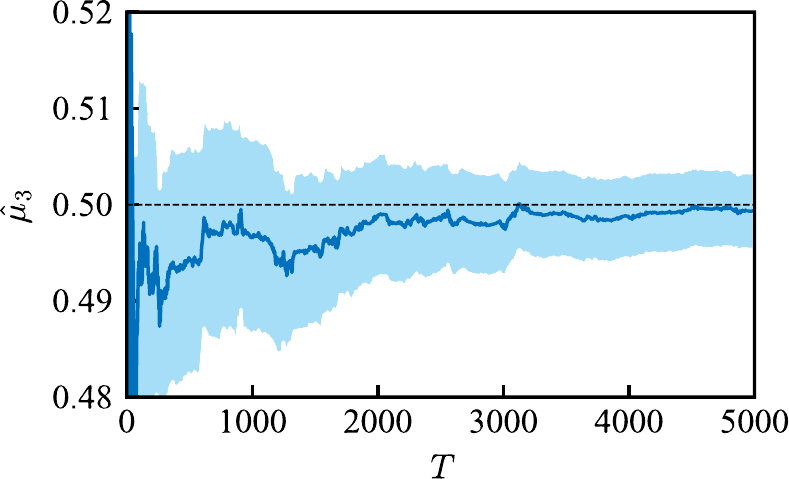}\label{fig:be_A12}}
    \subfigure[$\hat{A}_{21}=\hat{\mu}_{2}$.]{\includegraphics[scale=.5]{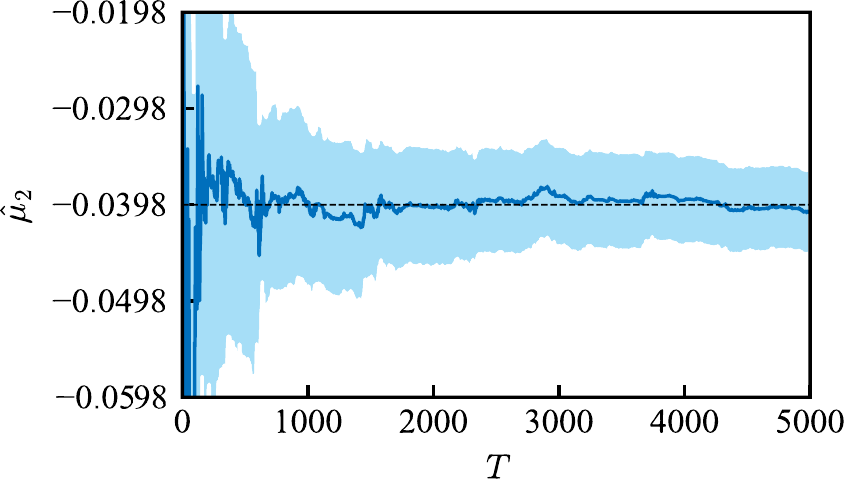}\label{fig:be_A21}}%
    \subfigure[$\hat{A}_{22}=\hat{\mu}_{4}$.]{\includegraphics[scale=.5]{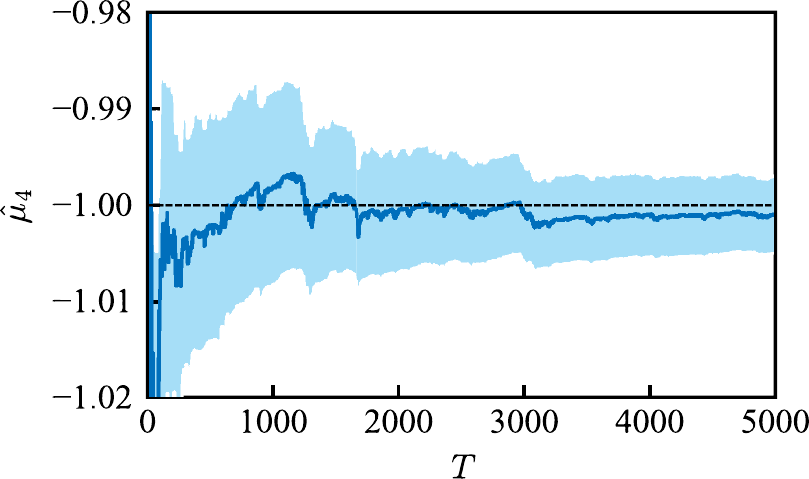}\label{fig:be_A22}}
    \caption{Result of Bayesian estimation for system matrix.}
    \label{fig:be}
\end{figure}

\begin{figure}[!t]
    \centering
    \includegraphics[scale=1]{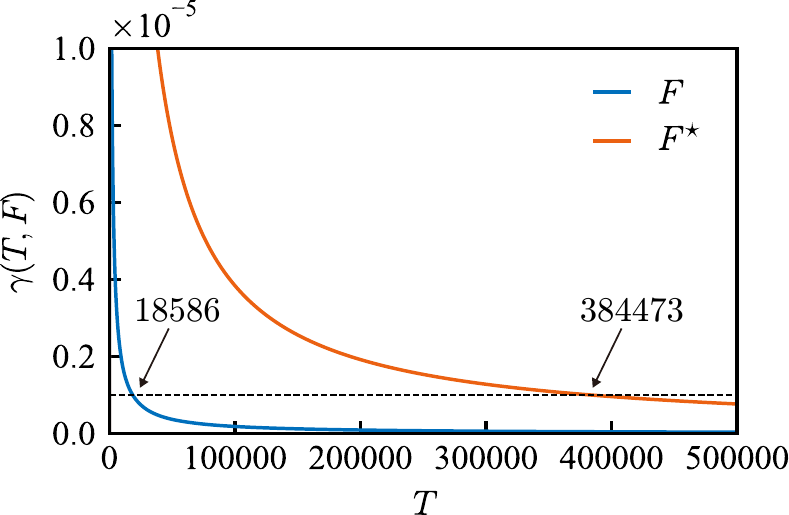}
    \caption{Comparison of sample identifying-complexity curves.}
    \label{fig:sicc}
\end{figure}

\begin{figure}[!t]
    \centering
    \includegraphics[scale=1]{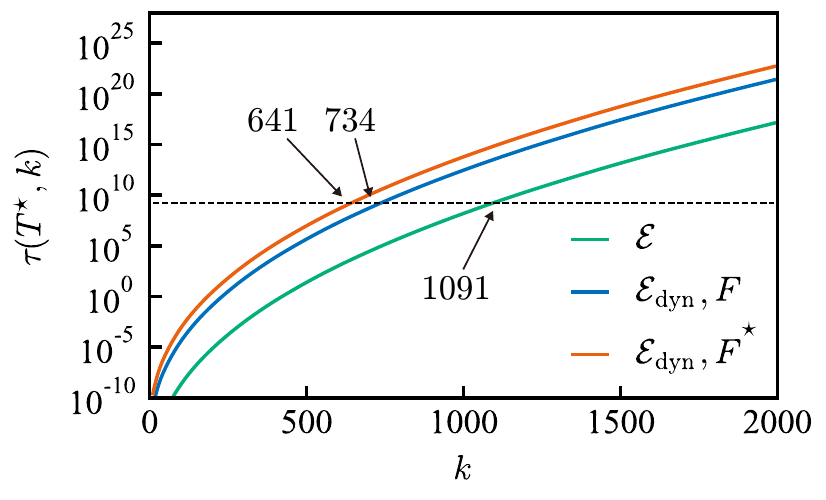}
    \caption{Comparison of sample deciphering-time curves.}
    \label{fig:sdtc}
\end{figure}

\begin{figure}[!t]
    \centering
    \includegraphics[scale=1]{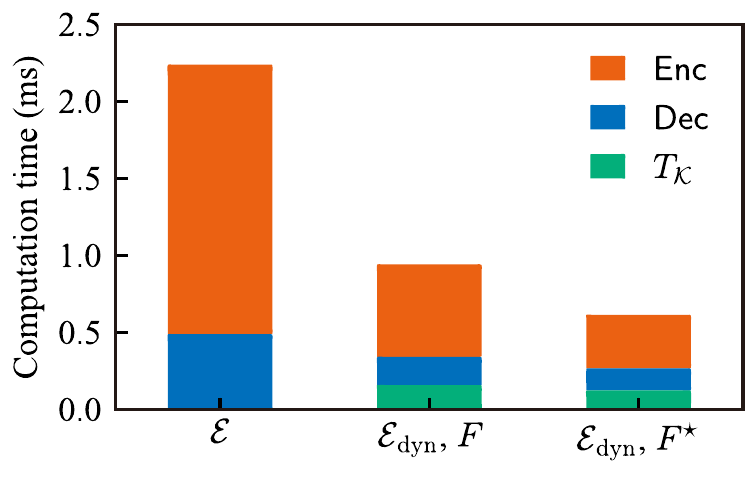}\label{fig:time}
    \caption{Comparison of average computation times on a plant side.}
    \label{fig:ct}
\end{figure}

\begin{table*}[!t]
    \centering
    \caption{Computation Times of $\Enc$, $\Dec$, $T_{\K}$, and $T_{\C}$ ($N=10000$)}
    \begin{tabular}{cccccccccccccccccc}
        \hline
       & \multirow{2}{*}{$k^{\star}$} & \multicolumn{4}{c}{$\Enc$ (ms)}   & \multicolumn{4}{c}{$\Dec$ (ms)}   & \multicolumn{4}{c}{$T_{\K}$ (ms)} & \multicolumn{4}{c}{$T_{\C}$ (ms)} \\
                             &        & Min    & Ave    & Max    & Std    & Min    & Ave    & Max    & Std    & Min    & Ave    & Max    & Std    & Min    & Ave    & Max    & Std    \\ \hline
        $\E$                 & $1091$ & $1.64$ & $1.75$ & $2.90$ & $0.09$ & $0.46$ & $0.49$ & $1.04$ & $0.04$ & --     & --     & --     & --     & --     & --     & --     & --     \\
        $\Edyn$, $F$         & $734$  & $0.56$ & $0.60$ & $1.25$ & $0.04$ & $0.17$ & $0.18$ & $0.57$ & $0.02$ & $0.15$ & $0.16$ & $0.52$ & $0.01$ & $0.42$ & $0.47$ & $1.00$ & $0.05$ \\
        $\Edyn$, $F^{\star}$ & $641$  & $0.32$ & $0.35$ & $0.78$ & $0.03$ & $0.13$ & $0.14$ & $0.41$ & $0.02$ & $0.11$ & $0.13$ & $0.33$ & $0.02$ & $0.32$ & $0.34$ & $0.82$ & $0.03$ \\ \hline
    \end{tabular}
    \vspace{-3mm}
    \label{tab:ct}
\end{table*}

\section{Conclusion}\label{sec:conclusion}

This paper addressed a systematic design of encrypted control systems aginst eavesdropping attacks to construct secure cyber-physical systems.
To quantify the security level of encrypted control systems, the novel security notions, sample identifying complexity and sample deciphering time, were proposed.
The sample identifying complexity characterizes the difficulty of system identification by means of a controllability gramian of a closed-loop system.
Additionally, the sample deciphering time represents the computation time for breaking ciphertexts to collect a data set for the identification.
Combining the notions, the optimal controller was obtained by the traditional stochastic cheap controller that simultaneously maximizes the stability degree of a closed-loop system and the difficulty of the identification.
Furtermore, the optimal key length was determined as the minimum key length enough to prevent the identification with a given precision within a life span of the system.
The numerical simulations demonstrated that the optimal key length and controller effectively reduced the implementation costs of encrypted control systems while keeping their security level.

In our best knowledge, this paper is the first work to reveal the relationship between the cryptographic security and dynamical systems in a control-theoretic manner.
One might think that some papers already related security level and properties of control systems~\cite{Dibaji19,Murguia20,Sandberg10,Milosevic20,Feng21,Cetinkaya20}.
However, these studies considered only the control-theoretic aspect of the impact of cyber-attacks, namely resilience, performance degradation, and detectability.
In contrast, our approach connected the effect of eavesdropping attacks and the characteristics of dynamical systems taking the feasibility of the attacks into consideration in terms of a computation time.

In this paper, the precision of adversary's estimation was evaluated based on a variance, i.e., the second moment about a mean.
However, we did not consider the first moment about the origin of the estimation.
In fact, although the estimator \eqref{eq:est_mean} is a consistent estimator, it is not a non-baiased estimator.
Hence, the adversary would obtain the estimates including a bias with a precision evaluated by the second moment.
This means the security evaluation of this paper is strict with a defender.
We will modify the proposed method to consider both the first and second moments.

Moreover, the estimation of system and input matrices of \eqref{eq:plant} rather than a system matrix of \eqref{eq:system} will be considered.
This would be achievable by rewriting \eqref{eq:plant} as
\[
    \begin{bmatrix}
        x_{t+1} \\
        \bfzero
    \end{bmatrix}=
    \begin{bmatrix}
        A_{p} & B_{p} \\
        O & O
    \end{bmatrix}
    \begin{bmatrix}
        x_{t} \\
        u_{t}
    \end{bmatrix}+
    \begin{bmatrix}
        w_{t} \\
        \bfzero
    \end{bmatrix}.
\]
The equation is the same form of \eqref{eq:system}, and thus, the discussions in this paper would be extended directly.
We will also consider extending the security concepts to be used for more general encrypted control systems, namely dynamic output-feedback controllers and nonlinear plants.
This can be achieved, for example, by using the input-output history feedback controller representation~\cite{Teranishi21} and the Koopman operator~\cite{Koopman31}.
The controller representation realizes a dynamic controller as a matrix-vector product form such as \eqref{eq:controller}.
The Koopman operator lifts a finite-dimensional nonlinear system to an infinite-dimensional linear system.
The proposed scheme can be applied for such systems up to an error due to the truncation of the system dimension.

\appendix

\subsection{Quantization in encrypted control}\label{app:quantization}

This section describes the properties of quantization errors in encrypted control systems with a linear controller
\begin{equation}
    f:(\Phi,\xi)\mapsto\psi=\Phi\xi,
    \label{eq:ctlr}
\end{equation}
where $\Phi\in\R^{\alpha\times\beta}$ is a controller parameter matrix, $\xi\in\R^{\beta}$ is a controller input vector, which consists of a controller state and plant output, and $\psi\in\R^{\alpha}$ is a controller output vector, which consists of a controller state update and plant input.
In this case, an encrypted controller of \eqref{eq:ctlr} with the ElGamal encryption $\E=(\Gen,\Enc,\Dec)$ is given as
\begin{equation}
    f_{\E}:(c_{\Phi},c_{\xi})\mapsto c_{\Psi}
    \label{eq:enc_ctlr}
\end{equation}
where $c_{\Phi}=\Enc(\pk,m_{\Phi})$, $c_{\xi}=\Enc(\pk,m_{\xi})$, $c_{\Psi_{ij}}=c_{\Phi_{ij}}\ast c_{\xi_{j}}\bmod p$, and $m_{\Phi}$ and $m_{\xi}$ are plaintexts of $\Phi$ and $\xi$, respectively.

For implementation of the encrypted controller \eqref{eq:enc_ctlr}, $\xi$ and $\Phi$ must be encoded to plaintexts before encryption and be decoded to real numbers after decryption using an encoder and decoder because a plaintext space is not a set of real numbers.
This study uses the following encoder $\Ecd_{\Delta}:\R\to\M$ and decoder $\Dcd_{\Delta}:\M\to\R$ with a sensitivity $\Delta>0$:
\begin{align*}
    \Ecd_{\Delta}&:x\mapsto\check{x}=\min\left\{\argmin_{m\in\M}|x/\Delta+p\ind{\R_{<0}}{x}-m|\right\}, \\
    \Dcd_{\Delta}&:\check{x}\mapsto\bar{x}=\Delta(\check{x}-p\ind{\Z_{>q}}{\check{x}}),
\end{align*}
where $\R_{<0}\coloneqq\{x\in\R\mid x<0\}$, $\Z_{>q}\coloneqq\{z\in\Z\mid q<z\}$, and $\ind{\mathcal{A}}{\cdot}$ is an indicator function that outputs $1$ if its argument belongs to a set $\mathcal{A}$; otherwise it outputs $0$.
$\Ecd_{\Delta}$ and $\Dcd_{\Delta}$ perform elementwise for a vector and matrix.

Define $\Qtz_{\Delta}\coloneqq\Dcd_{\Delta}\circ\Ecd_{\Delta}$, then $\Qtz_{\Delta}:x\mapsto\bar{x}=x+\tilde{x}$ can be regarded as a quantizer~\cite{Teranishi19_3}.
A quantization error $\tilde{x}=\Qtz_{\Delta}(x)-x$ of $x\in\R$ holds
\begin{equation}
    |\tilde{x}|\le\Delta d_{\max}/2\,,
    \label{eq:error_bound}
\end{equation}
where $d_{\max}$ is the maximum difference between two consecutive elements in the plaintext space.
The inequality \eqref{eq:error_bound} implies that the quantization error decreases as $\Delta$ decreases.
Actually, the following propositions on the relationship between quantization errors and a sensitivity  hold:

\begin{proposition}\label{prop:enc_cont_1}
    Given the ElGamal encryption $\E$ and $f_{\E}$ in \eqref{eq:enc_ctlr}.
    Let $(\pk,\sk)=\Gen(k)$, $m_{\Phi}=\Ecd_{\Delta_{\Phi}}(\Phi)$, $m_{\xi}=\Ecd_{\Delta_{\xi}}(\xi)$, $c_{\Phi}=\Enc(\pk,m_{\Phi})$, and $c_{\xi}=\Enc(\pk,m_{\xi})$.
    Suppose an overflow and underflow do not occur, that is, $\|\Qtz_{\Delta_{\Phi}}(\Phi)/\Delta_{\Phi}\|_{\max}\|\Qtz_{\Delta_{\xi}}(\xi)/\Delta_{\xi}\|_{\infty}\le q$.
    If $\Delta_{\Phi}\to0$ and $\Delta_{\xi}\to0$, then $f_{\E}$ satisfies
    \[
        f^{+}\circ\Dcd_{\Delta_{\Phi}\Delta_{\xi}}\circ\Dec(\sk,f_{\E}(c_{\Phi},c_{\xi}))=f(\Phi,\xi)=\psi,
    \]
    where $f=f^{+}\circ f^{\times}$, and the maps $f^{\times}$ and $f^{+}$ are a multiplication and addition part of $f$ in \eqref{eq:ctlr}, respectively~\cite{Kogiso15}.
\end{proposition}

\begin{proof}
    From the multiplicative homomorphism of $\E$, we have $f^{+}\circ\Dcd_{\Delta_{\Phi}\Delta_{\xi}}\circ\Dec(\sk,f_{\E}(c_{\Phi},c_{\xi}))=f^{+}\circ\Dcd_{\Delta_{\Phi}\Delta_{\xi}}\circ f^{\times}(m_{\Phi},m_{\xi})$.
    Regardless of the signs of $\Phi_{ij}$ and $\xi_{j}$, when an overflow and underflow do not occur, then the $(i,j)$ entry of $\Dcd_{\Delta_{\Phi}\Delta_{\xi}}\circ f^{\times}(m_{\Phi},m_{\xi})$ is given as $\Delta_{\Phi}\Delta_{\xi}(\Phi_{ij}/\Delta_{\Phi}+\delta_{\Phi_{ij}})(\xi_{j}/\Delta_{\xi}+\delta_{\xi_{j}})=\Psi_{ij}+\tilde{\Psi}_{ij}$, where $\tilde{\Psi}_{ij}=\tilde{\Phi}_{ij}\xi_{j}+\Phi_{ij}\tilde{\xi}_{j}+\tilde{\Phi}_{ij}\tilde{\xi}_{j}$, $\tilde{\Phi}_{ij}=\Delta_{\Phi}\delta_{\Phi_{ij}}=\Qtz_{\Delta_{\Phi}}(\Phi_{ij})-\Phi_{ij}$, and $\tilde{\xi}_{j}=\Delta_{\xi}\delta_{\xi_{j}}=\Qtz_{\Delta_{\xi}}(\xi_{j})-\xi_{j}$.
    From \eqref{eq:error_bound}, we obtain $|\tilde{\Psi}_{ij}|\le|\xi_{j}|\Delta_{\Phi}d_{\max}/2+|\Phi_{ij}|\Delta_{\xi}d_{\max}/2+\Delta_{\Phi}\Delta_{\xi}d_{\max}^{2}/4$.
    Therefore, $f^{+}\circ\Dcd_{\Delta_{\Phi}\Delta_{\xi}}\circ f^{\times}(m_{\Phi},m_{\xi})=\psi$ as $\Delta_{\Phi}\to0$ and $\Delta_{\xi}\to0$ because of $\tilde{\Psi}_{ij}\to0$.
\end{proof}

\begin{proposition}\label{prop:enc_cont_2}
    Let $k$ be a key length of the ElGamal encryption.
    Suppose $\max\{\|\Phi\|_{\max},\|\xi\|_{\infty}\}$ exists.
    Then, there exist seisitivities $\Delta_{\Phi}(k)$ and $\Delta_{\xi}(k)$ satisfying $\|\Qtz_{\Delta_{\Phi}}(\Phi)/\Delta_{\Phi}(k)\|_{\max}\|\Qtz_{\Delta_{\xi}}(\xi)/\Delta_{\xi}(k)\|_{\infty}\le q$ such that $\lim_{k\to\infty}\Delta_{\Phi}(k)=0$ and $\lim_{k\to\infty}\Delta_{\xi}(k)=0$.
\end{proposition}

\begin{proof}
    Let $\Delta(k)=\Delta_{\Phi}(k)=\Delta_{\xi}(k)$, $\bar{\Phi}=\Qtz_{\Delta_{\Phi}}(\Phi)$,  $\bar{\xi}=\Qtz_{\Delta_{\xi}}(\xi)$, and $D=\max\{\|\Phi\|_{\max},\|\xi\|_{\infty}\}$.
    Then, the inequality $\|\bar{\Phi}/\Delta_{\Phi}(k)\|_{\max}\|\bar{\xi}/\Delta_{\xi}(k)\|_{\infty}\le q$ can be deformed as $\Delta(k)\ge\sqrt{\|\bar{\Phi}\|_{\max}\|\bar{\xi}\|_{\infty}/q}$.
    Since $q\in(2^{k-1},2^{k})$, we obtain the sufficient condition $\Delta(k)=\sqrt{(D+\Delta(k)d_{\max}/2)^{2}/2^{k-1}}=D/(2^{\frac{1}{2}(k-1)}-d_{\max}/2)$ to hold the inequality, where, using \eqref{eq:error_bound}, $\|\bar{\Phi}\|_{\max}\|\bar{\xi}\|_{\infty}$ is bounded from above as $\|\bar{\Phi}\|_{\max}\|\bar{\xi}\|_{\infty}\le(\|\Phi\|_{\max}+\Delta(k) d_{\max}/2)(\|\xi\|_{\infty}+\Delta(k)d_{\max}/2)\le(D+\Delta(k)d_{\max}/2)^{2}$.
    By definition, every elements in a plaintext space of the ElGamal encryption are quadratic residues modulo $p$.
    The author of~\cite{Burgess57} shows that the number of consecutive quadratic non-residues modulo $p$ is at most $O(p^{\frac{1}{4}+\delta})$ for large $p$ and any positive number $\delta$.
    This means $d_{\max}\le2^{(\frac{1}{4}+o(1))(k+1)}$ since $p\in(2^{k},2^{k+1})$, and $\lim_{k\to\infty}2^{\frac{1}{2}(k-1)}-d_{\max}/2\ge\lim_{k\to\infty}2^{\frac{1}{2}(k-1)}-2^{(\frac{1}{4}+o(1))(k+1)}=\infty$.
    Therefore, $\Delta(k)=\Delta_{\Phi}(k)=\Delta_{\xi}(k)\to0$ as $k\to\infty$.
\end{proof}

\propref{prop:enc_cont_1} shows an output of $f_{\E}$ exactly matches one of $f^{\times}$ if the sensitivities are zero as long as $\Ecd_{\Delta}$ and $\Dcd_{\Delta}$ do not cause an overflow and underflow.
\propref{prop:enc_cont_2} guarantees such sensitivities exist when a key length is sufficiently large.
If a key length is relatively small, then quantization errors cannot be ignored.
The quantization errors would degrade the precision of adversary's estimation of $A$ in \eqref{eq:system}, that is, the number of data needs to be increased in order to keep the precision of the estimation.

\subsection{Technical lemmas}

This section introduces two technical lemmas used for proofs in the following appendices.

\begin{lemma}\label{lem:quadratic}
    Let $v\in\R^{n}$ and $M\in\R^{n\times n}$, then
    \[
        \tr((v\otimes I)M(v\otimes I)^{\top})=\tr(M)\|v\|^{2}.
    \]
\end{lemma}

\begin{proof}
    \begin{align*}
        &\tr((v\otimes I)M(v\otimes I)^{\top}) \\
        &=\tr{\left(
            \begin{bmatrix}
                v_{1}I \\
                \vdots \\
                v_{n}I
            \end{bmatrix}
            \begin{bmatrix}
                m_{11} & \cdots & m_{1n} \\
                \vdots & \ddots & \vdots \\
                m_{n1} & \cdots & m_{nn}
            \end{bmatrix}
            \begin{bmatrix}
                v_{1}I & \cdots & v_{n}I
            \end{bmatrix}
        \right)}\!, \\
        &=\tr(\diag{m_{11}v_{1}^{2},\cdots,m_{nn}v_{1}^{2},\cdots,m_{11}v_{n}^{2},\cdots,m_{nn}v_{n}^{2}}), \\
        &=(m_{11}+\cdots+m_{nn})(v_{1}^{2}+\cdots+v_{n}^{2})=\tr(M)\|v\|^{2}.
    \end{align*}
    This completes the proof.
\end{proof}

\begin{lemma}\label{lem:invertible}
    Let $M\in\R^{m\times n}$ be a full column rank matrix, and $P\in\R^{m\times m}$ be a positive definite matrix, then $M^{\top}PM$ is positive definite and invertible.
\end{lemma}

\begin{proof}
    For any non-zero vector $x\in\R^{n}$,
    \[
        x^{\top}(M^{\top}PM)x=(Mx)^{\top}P(Mx)>0.
    \]
    Therefore, $M^{\top}PM$ is positive definite, and this also means it is invertible.
\end{proof}

\subsection{Proof of \lemref{lem:gamma}}\label{proof:lem_gamma}

\begin{proof}
    From Bayes' theorem, the probability density functions in \defref{def:adversary} hold $p(A|\D)=p(\D|A)p(A)/p(\D)\propto p(\D|A)p(A)$.
    Additionally,
    \begin{align*}
        p(\D|A)
        &=p(x_{0})\prod_{t=0}^{T-1}p(x_{t+1}|x_{t}), \\
        &=p(x_{0})\prod_{t=0}^{T-1}f(x_{t+1};(x_{t}\otimes I)^{\top}\vec(A),L^{-1}).
    \end{align*}
    where $Ax_{t}=(x_{t}\otimes I)^{\top}\vec(A)$.
    It follows that
    \begin{align*}
        &\ln{p(A|\D)} \\
        &=\ln p(\D|A)+\ln p(A)+\const, \\
        &=\sum_{t=0}^{T-1}\ln f(x_{t+1};(x_{t}\otimes I)^{\top}\vec(A),L^{-1}) \\
        &\quad+\ln f(\vec(A);\mu,\Lambda^{-1})+\ln p(x_{0})+\const, \\
        &=-\cfrac{1}{2}\sum_{t=0}^{T-1}(x_{t+1}-(x_{t}\otimes I)^{\top}\vec(A))^{\top}L \\
        &\qquad\qquad\qquad\qquad\quad(x_{t+1}-(x_{t}\otimes I)^{\top}\vec(A)) \\
        &\quad-\cfrac{1}{2}(\vec(A)-\mu)^{\top}\Lambda(\vec(A)-\mu)+\const, \\
        &=-\cfrac{1}{2}\left\{\vec(A)^{\top}\!\left(\Lambda+\sum_{t=0}^{T-1}(x_{t}\otimes I)L(x_{t}\otimes I)^{\top}\right)\!\vec(A)\right. \\
        &\quad\left.-2\vec(A)^{\top}\!\left(\Lambda\mu+\sum_{t=0}^{T-1}(x_{t}\otimes I)Lx_{t+1}\right)\!\right\}+\const, \\
        &=-\cfrac{1}{2}\,(\vec(A)^{\top}\hat{\Lambda}\vec(A)-2\vec(A)^{\top}\hat{\Lambda}\hat{\mu})+\const,
    \end{align*}
    where
    \begin{align*}
        \hat{\Lambda}&=\Lambda+\sum_{t=0}^{T-1}(x_{t}\otimes I)L(x_{t}\otimes I)^{\top}, \\
        \hat{\mu}&=\hat{\Lambda}^{-1}\left(\Lambda\mu+\sum_{t=0}^{T-1}(x_{t}\otimes I)Lx_{t+1}\right).
    \end{align*}
    That is, $p(A|\D)=f(\vec(A);\hat{\mu},\hat{\Lambda}^{-1})$.
    Furthermore, it follows from \lemref{lem:quadratic} that
    \[
        \tr(\hat{\Lambda}^{-1})\ge n^{2}\tr(\hat{\Lambda})^{-1}=n^{2}\left(\tr(\Lambda)+\tr(L)\sum_{t=0}^{T-1}\|x_{t}\|^{2}\right)^{-1}.
    \]
    The solution of \eqref{eq:system} is $x_{t}=\sum_{i=0}^{t}A^{i}w_{t-i-1}$ for all $t\ge0$, where $w_{-1}\coloneqq x_{0}$.
    Therefore, the sample-identifying complexity curve $\gamma(T,F)$ is given as
    \begin{align*}
        \gamma(T,F)
        &\!=\!\EV{n^{2}\left(\tr(\Lambda)+\tr(L)\sum_{t=0}^{T-1}\|x_{t}\|^{2}\right)^{-1}}, \\
        &\!=\!n^{2}\left\{\tr(\Lambda)+\tr(L)\sum_{t=0}^{T-1}\tr\left(\sum_{i=0}^{t}A^{i}\Sigma(A^{i})^{\top}\right)\right\}^{-1},
    \end{align*}
    where
    \begin{align*}
        \EV{\|x_{t}\|^{2}}
        &=\tr\left(\!\EV{\!\left(\sum_{i=0}^{t}A^{i}w_{t-i-1}\right)\!\!\left(\sum_{i=0}^{t}A^{i}w_{t-i-1}\right)^{\top}}\!\right), \\
        &=\tr\left(\sum_{i=0}^{t}A^{i}\EV{w_{t-i-1}w_{t-i-1}^{\top}}(A^{i})^{\top}\right), \\
        &=\tr\left(\sum_{i=0}^{t}A^{i}\Sigma(A^{i})^{\top}\right).
    \end{align*}
    This completes the proof.
\end{proof}

\subsection{Proof of \corref{cor:est_error}}\label{proof:cor_est_error}

\begin{proof}
    When a sample size $T$ is sufficiently large, the estimates $\hat{\Lambda}$ and $\hat{\mu}$ are given by
    \begin{align*}
        \hat{\Lambda}&\approx\sum_{t=0}^{T-1}(x_{t}\otimes I)L(x_{t}\otimes I)^{\top}, \\
        \hat{\mu}&\approx\vec(A)+\hat{\Lambda}^{-1}\left(\sum_{t=0}^{T-1}(x_{t}\otimes I)Lw_{t}\right),
    \end{align*}
    where $x_{t+1}=Ax_{t}+w_{t}=(x_{t}\otimes I)^{\top}\vec(A)+w_{t}$.
    Hence,
    \begin{align*}
        &\EV{\|A-\hat{A}\|_{F}^{2}}=\EV{\|\vec(A)-\hat{\mu}\|^{2}} \\
        &=\EV{\left\|\hat{\Lambda}^{-1}\left(\sum_{t=0}^{T-1}(x_{t}\otimes I)Lw_{t}\right)\right\|^{2}}, \\
        &=\EV{\tr\!\left(\!\hat{\Lambda}^{-1}\!\left(\sum_{t=0}^{T-1}(x_{t}\!\otimes\! I)Lw_{t}w_{t}^{\top}L^{\top}(x_{t}\!\otimes\! I)^{\top}\!\right)\!\left(\!\hat{\Lambda}^{-1}\!\right)^{\!\top}\!\right)\!}, \\
        &=\EV{\tr\left(\hat{\Lambda}^{-1}\left(\sum_{t=0}^{T-1}(x_{t}\otimes I)L(x_{t}\otimes I)^{\top}\right)\left(\hat{\Lambda}^{-1}\right)^{\top}\right)}, \\
        &=\EV{\tr(\hat{\Lambda}^{-1})}.
    \end{align*}
    This completes the proof.
\end{proof}

\subsection{Proof of \corref{cor:gamma}}\label{proof:cor_gamma}

\begin{proof}
    From the assumptions, $\tr(L)=n\sigma^{-2}$ and
    \[
        \tr\left(\sum_{i=0}^{t}A^{i}\Sigma(A^{i})^{\top}\right)=\sigma^{2}\tr\left(\sum_{i=0}^{t}A^{i}(A^{i})^{\top}\right).
    \]
    Thus,
    \begin{align*}
        \gamma(T,F) 
        &=n^{2}\left\{n\sigma^{-2}\sum_{t=0}^{T-1}\sigma^{2}\tr\left(\sum_{i=0}^{t}A^{i}(A^{i})^{\top}\right)\right\}^{-1}, \\
        &=n\left\{\sum_{t=0}^{T-1}\tr\left(\sum_{i=0}^{t}A^{i}(A^{i})^{\top}\right)\right\}^{-1}.
    \end{align*}
    This completes the proof.
\end{proof}

\subsection{Proof of \propref{prop:correctness}}\label{proof:correctness}

\begin{proof}
    Let $(\pk_{t},\sk_{t})=((p,q,g,h_{t}),s_{t})$, $c_{t}=(c_{1,t},c_{2,t})$, and $r_{t}$ be a random number used in the encryption algorithm at time $t$.
    From the proof of Theorem~$1$ in~\cite{Teranishi20_5}, $\Dec(\sk_{t},\Enc(\pk_{t},m))=m\bmod p$ is satisfied.
    The remaining part $\Dec(\sk_{t},c_{t})=m\bmod p$ is obtained by direct calculation as
    \begin{align*}
        {c_{1,t}}^{-s_{t}}c_{2,t}
        &=(g^{r_{t-1}}g^{r'})^{-s_{t}}(g^{r_{t-1}}g^{r'})^{s'}m{h_{t-1}}^{r_{t-1}}{h_{t-1}}^{r'}, \\
        &=mg^{-(r_{t-1}+r')(s_{t-1}+s')}g^{(r_{t-1}+r')(s_{t-1}+s')}, \\
        &=m\bmod p.
    \end{align*}
    Furthermore,
    \begin{align*}
        \Dec(\sk_{t},c_{t}\!\ast\!\Enc(\pk_{t},m')\bmod p)
        &\!=\!(c_{1,t}g^{r_{t}})^{-s_{t}}c_{2,t}m'{h_{t}}^{r_{t}}, \\
        &\!=\!{c_{1,t}}^{-s_{t}}c_{2,t}m'g^{-s_{t}r_{t}}g^{s_{t}r_{t}}, \\
        &\!=\!mm'\bmod p.
    \end{align*}
    This completes the proof.
\end{proof}

\subsection{Proof of \propref{prop:negligible}}\label{proof:negligible}

\begin{proof}
    Let $(\pk_{t},\sk_{t})=((p,q,g,h_{t}),s_{t})$ and $c_{t}=(c_{1,t},c_{2,t})$.
    The adversary cannot calculate $\sk_{t+1}=\sk_{t}+s_{t}'$ and $\sk_{t-1}=\sk_{t-2}+s_{t-2}'$ even though he/she knows $\sk_{t}=\sk_{t-1}+s_{t-1}'$, $\pk_{t}=g^{s_{t-1}+s_{t-1}'}$, $c_{1,t}=g^{r_{t-1}+r_{t-1}'}$, $c_{2,t}=mg^{(r_{t-1}+r_{t-1}')(s_{t-1}+s_{t-1}')}$, $m={c_{1,t}}^{-\sk_{t}}c_{2,t}$, and $r_{t-1}+r_{t-1}'=\log_{g}c_{1,t}$ as long as $s_{t}'$, $s_{t-1}'$, and $s_{t-2}'$ are secret.

    $s_{t}$ is randomly updated, i.e., $s_{t}\sim\U(\Z_{q})$~\cite{Teranishi20_5}.
    Additionally, $h_{t}=g^{s_{t}}\bmod p\sim\U(\G)$ because $s_{t}\sim\U(\Z_{q})$ and $\G$ is isomorphic to $\Z_{q}$~\cite{Teranishi20_5}.
    Similarly, $c_{1,t}=c_{1,t-1}g^{r_{t-1}'}\bmod p\sim\U(\G)$ and $c_{2,t}={c_{1,t-1}}^{s_{t-1}'}c_{2,t-1}(g^{s_{t-1}'}h_{t-1})^{r_{t-1}'}\bmod p\sim\U(\G)$ since $r_{t}'\sim\U(\Z_{q})$, $g^{r_{t}'}\bmod p\sim\U(\G)$, $g^{s_{t}'}h_{t}\bmod p\in\G$ and $(g^{s_{t}'}h_{t})^{r_{t}'}\bmod p\sim\U(\G)$.
    These facts conclude samples $\{h_{t}\}_{t\in\mathcal{I}}$, $\{s_{t}\}_{t\in\mathcal{I}}$, $\{c_{1,t}\}_{t\in\mathcal{I}}$, and $\{c_{2,t}\}_{t\in\mathcal{I}}$ for any time span $\mathcal{I}\subset[0,\infty)$ are unbiased.
    Therefore, the best strategy for the adversary to estimate $\sk_{t+1}$ and $\sk_{t-1}$ is random sampling from $\Z_{q}$, that is, $\Pr(\hat{\sk}_{t+1}=\sk_{t+1})=\Pr(\hat{\sk}_{t-1}=\sk_{t-1})=q^{-1}$.

    Let $\epsilon(k)=2^{-(k-1)}$ with a key length $k>1$.
    Then, for every positive integers $c>0$, there exists $N\in\Z$ such that $\epsilon(k)<k^{-c}$ for all $k>N$ because $\epsilon(k)$ and $k^{-c}$ decrease monotonically for $k>1$ and satisfy $\lim_{k\to\infty}\epsilon(k)/k^{-c}=\lim_{k\to\infty}2k^{c}/2^{k}=0$.
    Therefore, $\epsilon(k)$ is negligible and satisfies $\Pr(\hat{\sk}_{t+1}=\sk_{t+1})=\Pr(\hat{\sk}_{t-1}=\sk_{t-1})=q^{-1}<\epsilon(k)$ since $q\in(2^{k-1},2^{k})$.
\end{proof}

\subsection{Proof of \propref{prop:IND-CPA}}\label{proof:IND-CPA}

The security of a cryptosystem is formally defined via a game between a challenger and an adversary~\cite{Shoup06}.
The IND-CPA game is described as follows:
1) The challenger generates a key pair and gives the public key to the adversary.
2) The adversary chooses two plaintexts based on his/her knowledge that is only the public key in this case and sends the plaintexts to the challenger.
3) The challenger randomly selects a plaintext from the given plaintexts and returns it to the adversary.
4) The adversary guesses which plaintext was encrypted.
This process can be formulated by using probabilistic polynomial-time algorithms, $A_{0}$ and $A_{1}$, as follows.
\begin{center}
    \fbox{
        \begin{minipage}{.9\columnwidth}
            $\mathsf{IND\mathchar`-CPA}$
            \begin{enumerate}
                \item $(\pk,\sk)=\Gen(k).$
                \item $(m_{0},m_{1},\sigma)=A_{0}(\pk).$
                \item $c=\Enc(\pk,m_{b}),\ b\sim\U(\{0,1\}).$
                \item $\hat{b}=A_{1}(c,\sigma).$
            \end{enumerate}
        \end{minipage}
    }
\end{center}
The cryptosystem satisfies IND-CPA if the challenger wins the game, that is, the adversary's advantage $|\Pr(b=\hat{b})-1/2|$ is negligible.
We now show the proof of \propref{prop:IND-CPA} by reducing the IND-CPA game of the dynamic ElGamal encryption to $\mathsf{IND\mathchar`-CPA}$ of the normal ElGamal encryption.

\begin{proof}
    Consider the IND-CPA game of the dynamic ElGamal encrypion denoted by $\mathsf{IND\mathchar`-CPA_{dyn}^{0}}$.
    \begin{center}
        \fbox{
            \begin{minipage}{.9\columnwidth}
                $\mathsf{IND\mathchar`-CPA_{dyn}^{0}}$
                \begin{enumerate}
                    \item $(\pk_{0},\sk_{0})=\Gen(k).$
                    \item $(m_{0},m_{1},\sigma)=A_{0}(\pk_{0}).$
                    \item $c_{0}=\Enc(\pk_{0},m_{b}),\ b\sim\U(\{0,1\}).$
                    \item Set $t\leftarrow 0$.
                    \item $(\pk_{t+1},\sk_{t+1})=T_{\K}(\pk_{t},\sk_{t}).$
                    \item $c_{t+1}=T_{\C}(c_{t}).$
                    \item Set $t\leftarrow t+1$, and repeat 5) to 7) as needed.
                    \item $\hat{b}=A_{1}(\{c_{n}\}_{n=0}^{t},\{\pk_{n}\}_{n=0}^{t-1},\{\sk_{n}\}_{n=0}^{t-1},\sigma).$
                \end{enumerate}
            \end{minipage}
        }
    \end{center}

    In the modified game $\mathsf{IND\mathchar`-CPA_{dyn}^{1}}$, fix the lines 5) and 6) of $\mathsf{IND\mathchar`-CPA_{dyn}^{0}}$ to $(\pk_{t+1},\sk_{t+1})=\Gen(k)$ and $c_{t+1}=\Enc(\pk_{t+1},m_{b})$, respectively.
    \begin{center}
        \fbox{
            \begin{minipage}{.9\columnwidth}
                $\mathsf{IND\mathchar`-CPA_{dyn}^{1}}$
                \begin{enumerate}
                    \item $(\pk_{0},\sk_{0})=\Gen(k).$
                    \item $(m_{0},m_{1},\sigma)=A_{0}(\pk_{0}).$
                    \item $c_{0}=\Enc(\pk_{0},m_{b}),\ b\sim\U(\{0,1\}).$
                    \item Set $t\leftarrow 0$.
                    \item $(\pk_{t+1},\sk_{t+1})=\Gen(k).$
                    \item $c_{t+1}=\Enc(\pk_{t+1},m_{b}).$
                    \item Set $t\leftarrow t+1$, and repeat 5) to 7) as needed.
                    \item $\hat{b}=A_{1}(\{c_{n}\}_{n=0}^{t},\{\pk_{n}\}_{n=0}^{t-1},\{\sk_{n}\}_{n=0}^{t-1},\sigma).$
                \end{enumerate}
            \end{minipage}
        }
    \end{center}
    From the proof of \propref{prop:negligible}, this modification does not change any probability in $\mathsf{IND\mathchar`-CPA_{dyn}^{0}}$ since operations of $T_{\K}$ and $T_{\C}$ are completely random.
    Therefore, the difference between the adversary's advantages in $\mathsf{IND\mathchar`-CPA_{dyn}^{0}}$ and $\mathsf{IND\mathchar`-CPA_{dyn}^{1}}$ is negligible.
    Furthermore, the modification concludes $\{c_{n}\}_{n=0}^{t-1}$, $\{\pk_{n}\}_{n=0}^{t-1}$ and $\{\sk_{n}\}_{n=0}^{t-1}$ give no information about $c_{t}$ and $m_{b}$.
    Thus, we obtain the equivalent game of $\mathsf{IND\mathchar`-CPA_{dyn}^{1}}$, which is denoted by $\mathsf{IND\mathchar`-CPA_{dyn}^{2}}$.
    \begin{center}
        \fbox{
            \begin{minipage}{.9\columnwidth}
                $\mathsf{IND\mathchar`-CPA_{dyn}^{2}}$
                \begin{enumerate}
                    \item $(\pk_{0},\sk_{0})=\Gen(k).$
                    \item $(m_{0},m_{1},\sigma)=A_{0}(\pk_{0}).$
                    \item $c_{0}=\Enc(\pk_{0},m_{b}),\ b\sim\U(\{0,1\}).$
                    \item Set $t\leftarrow 0$.
                    \item $(\pk_{t+1},\sk_{t+1})=\Gen(k)$
                    \item $c_{t+1}=\Enc(\pk_{t+1},m_{b}).$
                    \item Set $t\leftarrow t+1$, and repeat 5) to 7) as needed.
                    \item $\hat{b}=A_{1}(c_{t},\sigma).$
                \end{enumerate}
            \end{minipage}
        }
    \end{center}

    $\mathsf{IND\mathchar`-CPA_{dyn}^{2}}$ is clearly the same as $\mathsf{IND\mathchar`-CPA}$ because the repetition of 5) to 7) does not affect the adversary's advantage.
    This fact concludes that the difference between adversary's advantages in $\mathsf{IND\mathchar`-CPA_{dyn}^{0}}$ and $\mathsf{IND\mathchar`-CPA}$ is negligible.
    In addition, the advantage in $\mathsf{IND\mathchar`-CPA}$ of the ElGamal encryption is negligible under the DDH assumption.
    From the above discussions, the dynamic ElGamal encryption satisfies IND-CPA at time $t$ under the DDH assumption.
\end{proof}

\subsection{Proof of \lemref{lem:opt_input}}\label{proof:opt_input}

\begin{proof}
    The problem that minimizes $J_{T}$ in \eqref{eq:cost} is a form of the typical finite-horizon discrete-time stochastic linear quadratic regulator problem~\cite{Kirk04}.
    Hence, the optimal control sequence is given as $u_{t}=-(B_{p}^{\top}P_{t+1}B_{p})^{-1}B_{p}^{\top}P_{t+1}A_{p}x_{t}$ if $B_{p}^{\top}P_{t+1}B_{p}$ is invertible, where $P_{t}=A_{p}^{\top}P_{t+1}A_{p}-A_{p}^{\top}P_{t+1}B_{p}(B_{p}^{\top}P_{t+1}B_{p})^{-1}B_{p}^{\top}P_{t+1}A_{p}+I$, and $P_{T}=I$.

    $P_{T}$ and $P_{T-1}=A_{p}^{\top}\{I-B_{p}(B_{p}^{\top}B_{p})^{-1}B_{p}^{\top}\}A_{p}+I$ are positive definite because $B_{p}^{\top}B_{p}$ is invertible from \lemref{lem:invertible}, and $B_{p}(B_{p}^{\top}B_{p})^{-1}B_{p}^{\top}$ is a hat matrix whose eigenvalues consist of $n$ ones and $m-n$ zeros.
    Assume that $P_{T-i}>0$ for $1\le i<T$, then $P_{T-i-1}=A_{p}^{\top}\{P_{T-i}-P_{T-i}B_{p}(B_{p}^{\top}P_{T-i}B_{p})^{-1}B_{p}^{\top}P_{T-i}\}A_{p}+I$, where $B_{p}^{\top}P_{T-i}B_{p}$ is positive definite and invertible from \lemref{lem:invertible}.
    Additionally,
    \[
        \begin{bmatrix}
            P_{T-i} & P_{T-i}B_{p} \\
            B_{p}^{\top}P_{T-i} & B_{p}^{\top}P_{T-i}B_{p}
        \end{bmatrix}>0
    \]
    since $P_{T-i}>0$.
    Thus, from Schur complement, $P_{T-i}-P_{T-i}B_{p}(B_{p}^{\top}P_{T-i}B_{p})^{-1}B_{p}^{\top}P_{T-i}>0$.
    Therefore, $P_{T-i-1}$ is positive definite, and $B_{p}^{\top}P_{t+1}B_{p}$ is invertible for all $t\in[0,T)$.
\end{proof}

\subsection{Proof of \propref{prop:sum_tr_Wt}}\label{proof:sum_tr_Wt}

\begin{proof}
    The cases of $T=1,2$ can be easily confirmed.
    For $T\ge3$, it follows that
    \begin{align*}
        E(T)
        &\!=\!E(T-1)+\sum_{i=0}^{T-2}\tr(A^{i}(A^{i})^{\top})+\tr(A^{T-1}(A^{T-1})^{\top}), \\
        &\!=\!E(T\!-\!1)\!+\!(E(T\!-\!1)\!-\!E(T\!-\!2))\!+\!\tr(A^{T\!-\!1}(A^{T\!-\!1})^{\!\top}), \\
        &\!=\!2E(T-1)-E(T-2)+\tr(A^{T-1}(A^{T-1})^{\top}).
    \end{align*}
    This completes the proof.
\end{proof}

\subsection{Proof of \thmref{thm:design}}\label{proof:design}

\begin{proof}
    From \eqref{eq:gamma_simple}, \eqref{eq:gramian}, and \eqref{eq:T_star}, we have the minimum time step $T^{\star}$ satisfying $\gamma(T^{\star},F^{\star})<\gamma_{c}$.
    Then, from \eqref{eq:tau} and \eqref{eq:k_star}, the key length $k^{\star}$ fulfills $\tau(T^{\star},k^{\star})>\tau_{c}$.
    Therefore, the encrypted control system is secure.
\end{proof}

\bibliographystyle{IEEEtran}
\bibliography{encrypted_control_and_optimization}

\begin{IEEEbiography}[{\includegraphics[width=1in,height=1.25in,clip,keepaspectratio]{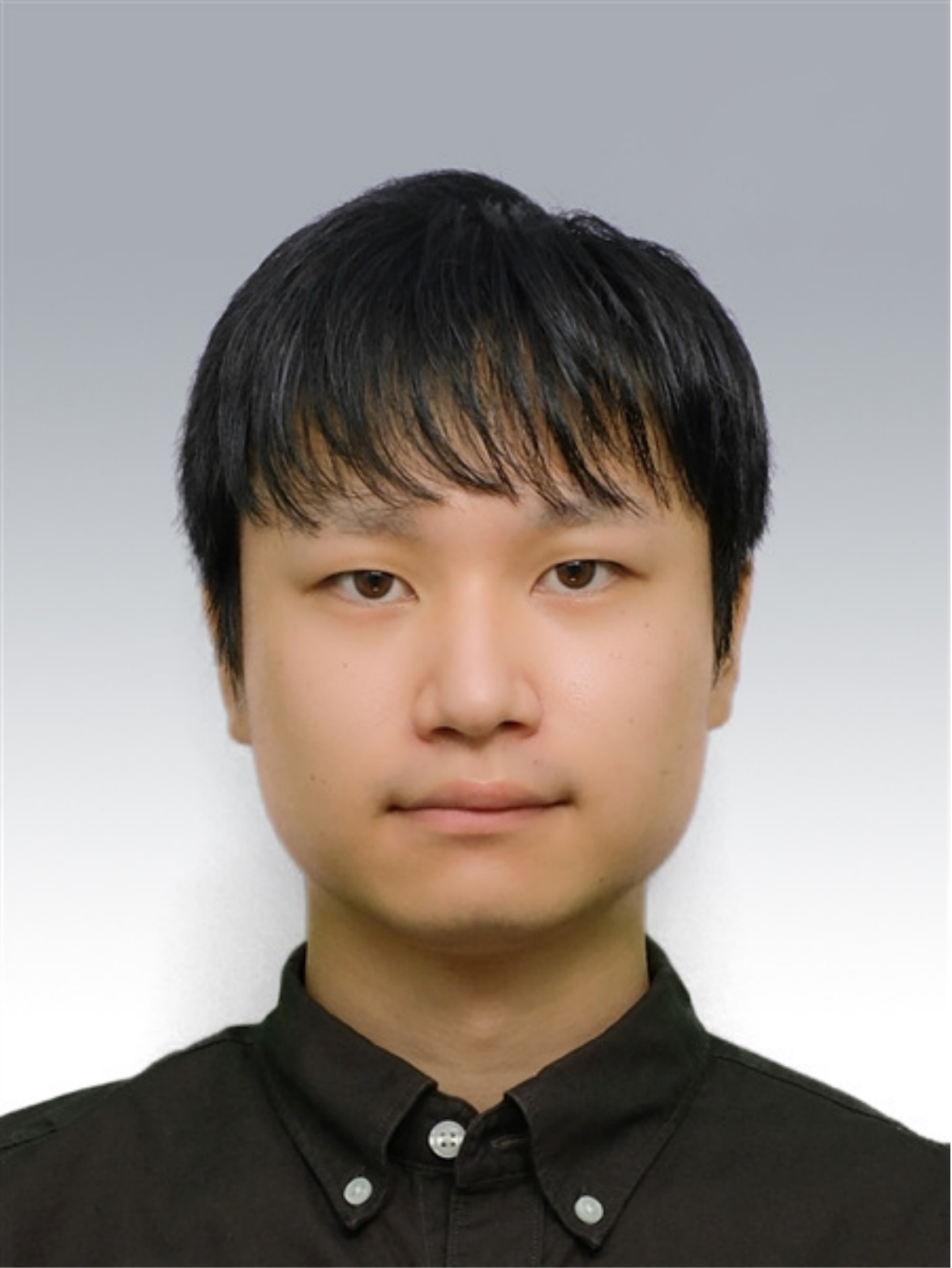}}]
    {Kaoru Teranishi} received the B.S. degree in electromechanical engineering from National Institute of Technology, Ishikawa College, Ishikawa, Japan, in 2019.
    He also obtained the M.S. degree in Mechanical and Intelligent Systems Engineering from The University of Electro-Communications, Tokyo, Japan, in 2021.
    He is currently a Ph.D. student at The University of Electro-Communications.
    From October 2019 to September 2020, he was a visiting scholar of the Georgia Institute of Technology, GA, USA.
    Since April 2021, he has been a Research Fellow of Japan Society for the Promotion of Science.
    His research interests include control theory and cryptography for cyber-security of control systems.
\end{IEEEbiography}

\begin{IEEEbiography}[{\includegraphics[width=1in,height=1.25in,clip,keepaspectratio]{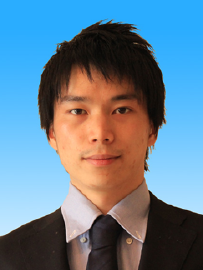}}]
    {Tomonori Sadamoto} received the Ph.D. degree from the Tokyo Institute of Technology, Tokyo, Japan in 2015. From June in 2015 to March in 2016, he was a Visiting Researcher at School of Electrical Engineering, Royal Institute of Technology, Stockholm, Sweden. From April 2016 to August 2016, he was a researcher with the Department of Systems and Control Engineering, Graduate School of Engineering, Tokyo Institute of Technology. From August 2016 to November 2018, he was a specially appointed Assistant Professor with the same department. Since November 2018, he has been assistant professor with Department of Mechanical and Intelligent Systems Engineering in the University of Electro-Communications. He was named as a finalist of the European Control Conference Best Student-Paper Award in 2014. He received Research encouragement award from The Funai Foundation for Informaiton Technology in 2019, and received IEEE Control Systems Magazine Outstanding Paper Award in 2020. 
\end{IEEEbiography}

\begin{IEEEbiography}[{\includegraphics[width=1in,height=1.25in,clip,keepaspectratio]{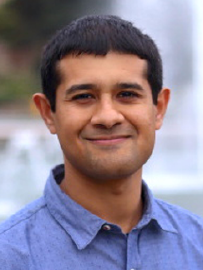}}]
    {Aranya Chakrabortty} received the Ph.D. degree in Electrical Engineering from Rensselaer Polytechnic Institute, NY in 2008. From 2008 to 2009 he was a postdoctoral research associate at University of Washington, Seattle, WA. From 2009 to 2010 he was an assistant professor at Texas Tech University, Lubbock, TX. Since 2010 he has joined the Electrical and Computer Engineering department at North Carolina State University, Raleigh, NC, where he is currently a Professor. His research interests are in all branches of control theory with applications to electric power systems. He received the NSF CAREER award in 2011. He was named as a University Faculty Scholar by the NC State Provost's office in 2019.
\end{IEEEbiography}

\begin{IEEEbiography}[{\includegraphics[width=1in,height=1.25in,clip,keepaspectratio]{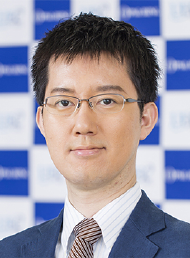}}]
    {Kiminao Kogiso} received the B.S., M.S., and Ph.D. degrees in mechanical engineering from Osaka University, Japan, in 1999, 2001, and 2004, respectively.
    
    He was a postdoctoral researcher in the 21st Century COE Program in 2004 and became an Assistant Professor in the Department of Information Systems, Nara Institute of Science and Technology, Nara, Japan, in 2005.
    Since March 2014, he has been an Associate Professor in the Department of Mechanical and Intelligent Systems Engineering, The University of Electro-Communications, Tokyo, Japan.
    From November 2010 to December 2011, he was a visiting scholar of the Georgia Institute of Technology, GA, USA.
    His research interests include constrained control, control of decision makers, cyber-security of control systems, and their applications.
\end{IEEEbiography}

\end{document}